%% file: main.tex
\documentclass[runningheads]{llncs}
\RequirePackage[colorlinks=true]{hyperref}
\hypersetup{
  linkcolor=[rgb]{0,0,0.4},
  citecolor=[rgb]{0, 0.4, 0},
  urlcolor=[rgb]{0.6, 0, 0}
}
\usepackage{todonotes}
\usepackage{pgfplots}
\usepackage{amsmath}
\usepackage{amsfonts}
\usepackage{lipsum}
\usepackage{setspace}
\usepackage{mdframed}
\usepackage{multicol}
\usepackage{multirow}
\usepackage{subcaption}
\usepackage{booktabs}
\usepackage{graphicx}
\usepackage{comment}
\usepackage{tikz}
\usepackage{xspace}
\usepgflibrary{arrows}

\usetikzlibrary{decorations.pathreplacing,backgrounds}
\usetikzlibrary{plotmarks}

\usepackage[font=footnotesize]{caption}

\usepackage{algorithmicx}
\usepackage{algorithm} % http://ctan.org/pkg/algorithms
\usepackage[noend]{algpseudocode}

\algrenewcommand\algorithmicindent{1.0em}%

\newcommand\sse{\subseteq}
\newcommand\Sym[1]{\ensuremath{\mathrm{Sym}_{#1}}}
\newcommand\set[1]{\ensuremath{\{#1\}}}
\newcommand\condset[2]{\set{#1 \;|\; #2}}
\newcommand\NP{\ensuremath{\mathsf{NP}}}
\newcommand\coNP{\ensuremath{\mathsf{coNP}}}
\newcommand\FF{\ensuremath{\mathcal{F}}}

\newcommand\poly[1]{\ensuremath{\mathrm{poly}(#1)}}

\newcommand{\isom}{\cong}
\newcommand{\puz}[2]{$(#1, #2)$-puzzle\xspace}
\newcommand{\puzs}[2]{$(#1, #2)$-puzzles\xspace}
\newcommand{\susp}[2]{$(#1, #2)$-SUSP\xspace}
\newcommand{\susps}[2]{$(#1, #2)$-SUSPs\xspace}

\date{}

\title{Matrix Multiplication:\\ Verifying Strong Uniquely Solvable
  Puzzles\thanks{An extended abstract of this paper appeared in the
    Proceedings of SAT 2020 \cite{ajx20}.}}
\author{
  Matthew Anderson%\inst{1}
  \and%
  Zongliang Ji%\inst{1}
  \and%
  Anthony Yang Xu%\inst{1}
}
\institute{Department of Computer Science \\ Union College
  \\ Schenectady, New York, USA \\ \email{\{andersm2, jiz,
    xua\}@union.edu}}

\begin{document}
\maketitle
\begin{abstract}
Cohn and Umans proposed a framework for developing fast matrix
multiplication algorithms based on the embedding computation in
certain groups algebras \cite{cu03}.  In subsequent work with
Kleinberg and Szegedy, they connected this to the search for
combinatorial objects called strong uniquely solvable puzzles (strong
USPs) \cite{cksu05}.  We begin a systematic computer-aided search for
these objects.  We develop and implement constraint-based algorithms
build on reductions to $\mathrm{SAT}$ and $\mathrm{IP}$ to verify that
puzzles are strong USPs, and to search for large strong USPs.  We
produce tight bounds on the maximum size of a strong USP for width $k
\le 5$, construct puzzles of small width that are larger than previous
work, and improve the upper bounds on strong USP size for $k \le 12$.
Although our work only deals with puzzles of small-constant width, the
strong USPs we find imply matrix multiplication algorithms that run in
$O(n^\omega)$ time with exponent $\omega \le 2.66$.  While our
algorithms do not beat the fastest algorithms, our work provides
evidence and, perhaps, a path to finding families of strong USPs that
imply matrix multiplication algorithms that are more efficient than
those currently known.

\end{abstract}

%\thispagestyle{empty}
%\newpage
%\pagenumbering{arabic}

\input{intro.tex}
\input{prelim.tex}

\input{verify.tex}
\input{search.tex}

\input{upper_bounds.tex}

\input{implementation.tex}

\input{results.tex}

\input{conclusion.tex}
\input{acknowledgments.tex}

\bibliographystyle{splncs04}
\bibliography{references}

\appendix

\input{appendix.tex}
\end{document}

%% file: intro.tex
\section{Introduction}
\label{sec:intro}

An optimal algorithm for matrix multiplication remains elusive despite
substantial effort.  We focus on the square variant of the matrix
multiplication problem, i.e., given two $n$-by-$n$ matrices $A$ and
$B$ over a field $\FF$, the goal is to compute the matrix product $C =
A \times B$.  The outstanding open question is: How many field
operations are required to compute $C$?  The long thought-optimal
na\"{i}ve algorithm based on the definition of matrix product is
$O(n^3)$ time.  The groundbreaking work of Strassen showed that it can
be done in time $O(n^{2.808})$ \cite{str69} using a divide-and-conquer
approach.  A long sequence of work concluding with Coppersmith and
Winograd's algorithm (CW) reduced the running time to $O(n^{2.376})$
\cite{pan78,sch81,str86,cw90}. Recent computer-aided refinements of CW
by others reduced the exponent to $\omega \le 2.3728639$
\cite{ds13,wil12,le14}.

\subsubsection{Approach}
Cohn and Umans \cite{cu03} introduced a framework for developing
faster algorithms for matrix multiplication by reducing this to a
search for groups with subsets that satisfy an algebraic property
called the \emph{triple-product property}, which allows matrix
multiplication to be embedded in the group algebra.  Their approach
takes inspiration from the $O(n \log n)$ algorithm for multiplying
degree-$n$ univariate polynomials by embedding into the group algebra
of the fast Fourier transform, c.f., e.g., \cite[Chapter 30]{clrs}.
Subsequent work \cite{cksu05} elaborated on this idea and developed
the notion of combinatorial objects called \emph{strong uniquely
  solvable puzzles} (strong USPs).  These objects imply a group
algebra embedding for matrix multiplication, and hence give a matrix
multiplication algorithm as well.
\begin{figure}[t]
\centering
  \includegraphics[width=\linewidth]{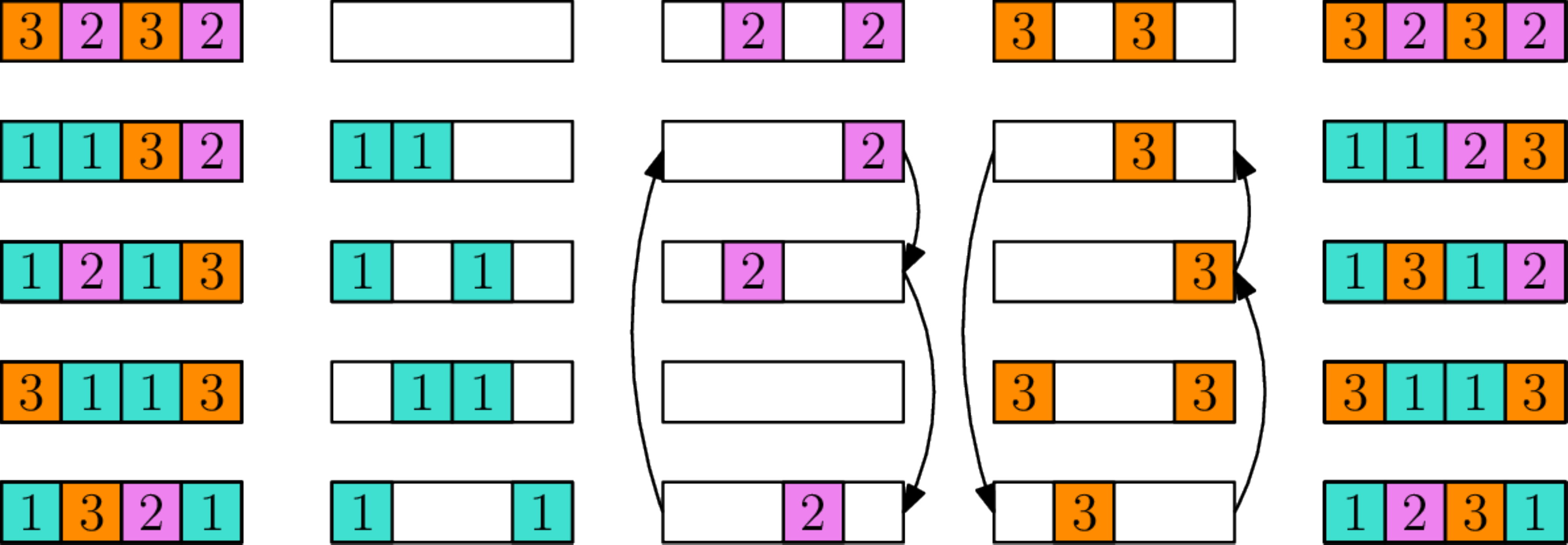}
\caption{The leftmost diagram is a width-4 size-5 puzzle $P$.  The
  middle three diagrams are the three sets of subrows of $P$. The
  rightmost diagram is the puzzle $P'$ resulting from reordering the
  subrows of $P$ as indicated by the arrows and then recombining them.
  Since $P$ can be rearranged as $P' \neq P$ without overlap, $P$ is
  not uniquely solvable.}\label{fig:puzzle}
\end{figure}

A \emph{width}-$k$ puzzle $P$ is a subset of $\set{1,2,3}^k$, and the
cardinality of $P$ is the puzzle's \emph{size}.  Each element of $P$
is called a \emph{row} of $P$, and each row consists of three
\emph{subrows} that are elements of $\set{1,*}^k$, $\set{2,*}^k$,
$\set{3,*}^k$ respectively.  Informally, a puzzle $P$ is a
\emph{uniquely solvable puzzle} (USP) if there is no way to permute
the subrows of $P$ to form a distinct puzzle $P'$ without cells with
numbers overlapping.  \autoref{fig:puzzle} demonstrates a puzzle
that is not a USP.  A uniquely solvable puzzle is \emph{strong} if a
tighter condition for non-overlapping holds (see
\autoref{def:strong-USP}). For a fixed width $k$, the larger
the size of a strong USP, the faster matrix multiplication algorithm
it gives \cite{cksu05}.  In fact, Cohn et al.~show that there exist an
infinite family of strong USPs that achieves $\omega < 2.48$.

We follow Cohn et al.'s program by developing: (i)
\textbf{verification algorithms} and heuristics to determine whether a
puzzle is a strong USP, (ii) \textbf{search algorithms} to find large
strong USPs, (iii) \textbf{practical implementations}\footnote{Source
  code available here: \url{https://bitbucket.org/paraphase/matmult}}
of these algorithms, and (iv) new \textbf{upper bounds} on the size of
strong USPs.  The most successful of our verification algorithms work
by reducing the problem through 3D matching to the satisfiability
($\mathrm{SAT}$) and integer programming ($\mathrm{IP}$) problems that
are then solved with existing tools.  The algorithms we develop are
not efficient---they run in worst-case exponential time in the natural
parameters.  However, the goal is to find a sufficiently large strong
USP that would provide a faster matrix multiplication algorithm, and
the resulting algorithm's running time is independent of the running
time of our algorithms.  The inefficiency of our algorithms limit the
search space that we can feasibly examine.

\subsubsection{Results}
Our theoretical results and implementation produces new bounds on the
size of the largest strong USP for small-width puzzles.  For
small-constant width, $k \le 12$, we beat the largest sizes of
\cite[Proposition 3.8]{cksu05}.  Our lower bounds on maximum size are
witnessed by strong USPs we found via search.  For $k \le 5$ we give
tight upper bounds determined by exhaustively searching all puzzles
after modding out common symmetries.  For $k \le 12$, we improve the
upper bounds on the size of strong USPs.  Although our current results
do not beat \cite{cksu05} for unbounded $k$, they give evidence that
there may exist families of strong USPs that give matrix
multiplication algorithms that are more efficient than those currently
known.  The best strong USP we can produce imply matrix multiplication
algorithms with $\omega \le 2.66$.

We also create a benchmark data set of SAT/UNSAT instances based on
our reductions from strong-USP verification and examine the
performance of solvers from the 2021 SAT Competition \cite{SAT2021}.

\subsubsection{Related Work}
For background on algorithms matrix multiplication problem, c.f, e.g.,
\cite{gs005}.  There are also a number of negative results known.
Na\"{i}vely, the dimensions of the output matrix $C$ implies that the
problem requires at least $\Omega(n^2)$ time.  Slightly better lower
bounds are known in general and also for specialized models of
computation, c.f., e.g., \cite{shp03,kam05}.  There are also lower
bounds known for a variety of algorithmic approaches to matrix
multiplication.  Ambainis et al.~showed that the laser method cannot
alone achieve an algorithm with $\omega \le 2.3078$ \cite{afl15}.  A
recent breakthrough on arithmetic progressions in cap sets
\cite{clp17} combined with a conditional result on the
Erd\"{o}s-Szemeredi sunflower conjecture \cite{asu13} imply that Cohn
et al.'s strong USP approach cannot achieve $\omega = 2 + \epsilon$
for some $\epsilon > 0$ \cite{bccgu16}.  Subsequent work has
generalized this barrier \cite{avw18a,avw18b} to a larger class of
algorithmic techniques.  Despite this, we are unaware of a concrete
lower bound on $\epsilon$ implied by these negative results.  There
remains a substantial gap in our understanding between what has been
achieved by the positive refinements of LeGall, Williams, and
Stothers, and the impossibility of showing $\omega = 2$ using the
strong USP approach.

Recently Fawzi et al.~showed how reinforcement learning techniques can
be used to develop new matrix multiplication algorithms
\cite{fawzi2022discovering}.  Their work produces matrix
multiplication algorithms with $\omega < 2.77$, which is faster than
Strassen's original algorithm ($\omega < 2.81$), but far from the
refinements of Coppersmith-Winograd ($\omega < 2.372$) or the results
achieved in this work.

\subsubsection{Organization}
\autoref{sec:prelim} begins with the formal definition of a strong USP
and the Cohn-Umans framework.  Sections~\ref{sec:verify} \&
\ref{sec:search}, respectively, discuss our algorithms and heuristics
for verifying that and searching for a puzzle that is a strong USP.
\autoref{sec:upper-bounds} describes several upper bounds on the size
of strong USPs. Sections~\ref{sec:implementation} \& \ref{sec:results}
discuss our implementation and experimental results.

%% file: prelim.tex
\section{Preliminaries}
\label{sec:prelim}

\newcommand\ordset[1]{\ensuremath{[#1]}}

For an integer $k$, we use $\ordset{k}$ to denote the set
$\set{1,2, \ldots, k}$.  For a set $Q$, $\Sym{Q}$ denotes the
symmetric group on the elements of $Q$, i.e., the group of
permutations acting on $Q$. Cohn et al.~introduced the idea of a
\emph{puzzle} \cite{cksu05}.
\begin{definition}[Puzzle]
  \label{def:puzzle}
  For $s, k \in \mathbb{N}$, an $(s,k)$-\emph{puzzle} is a subset $P \sse
  \ordset{3}^k$ with $|P| = s$.  We call $s$ the \emph{size} of $P$,
  and $k$ the \emph{width} of $P$.
\end{definition}
\noindent We say that an $(s,k)$-puzzle has $s$ rows and $k$ columns.
The columns of a puzzle are inherently ordered and indexed by
$[k]$. The rows of a puzzle have no inherent ordering, however, it is
often convenient to assume that they are ordered and indexed by the
set of natural numbers $[s]$.  

Cohn et al.~establish a particular combinatorial property of puzzles
that allows one to derive group algebras that matrix multiplication
can be efficiently embedded into.  Such puzzles are called
\emph{strong uniquely solvable puzzles}.  However, to give some
intuition we first explain a simpler version of the property called
\emph{uniquely solvable puzzles}.

\begin{definition}[Uniquely Solvable Puzzle (USP)]
  \label{def:USP}
  An \puz{s}{k} $P$ is \emph{uniquely solvable} if for all $\pi_1,
  \pi_2, \pi_3 \in \Sym{P}$: Either (i) $\pi_1 = \pi_2 = \pi_3$, or
  (ii) there exists $r \in P$ and $c \in \ordset{k}$ such that at
  least two of the following hold: $(\pi_1(r))_c = 1$, $(\pi_2(r))_c =
  2$, $(\pi_3(r))_c = 3$.
\end{definition}

Informally, a puzzle is \textbf{not} uniquely solvable if each row of
the puzzle can be broken into ones, twos, and threes pieces and then
the rows can be reassembled in a different way so that each new row is
a combination of a ones, a twos, and a threes piece where there is
exactly one element of $\ordset{3}$ for each column.  Observe that
uniquely solvable puzzles can have at most $2^k$ rows because each
ones piece, twos piece, and threes piece must be unique, as otherwise
the duplicate pieces can be swapped making the puzzle not uniquely
solvable.

The definition of \emph{strong} uniquely solvable puzzle is below, it
is nearly the same except that it requires that there be a collision
on a column between exactly two pieces, not two or more pieces like in
the original definition.
\begin{definition}[Strong USP (SUSP)]
  \label{def:strong-USP}
  An \puz{s}{k} $P$ is \emph{strong uniquely solvable} if for all
  $\pi_1, \pi_2, \pi_3 \in \Sym{P}$: Either (i) $\pi_1 = \pi_2 = \pi_3$,
  or (ii) there exists $r \in P$ and $c \in \ordset{k}$ such that exactly
  two of the following hold: $(\pi_1(r))_c = 1$, $(\pi_2(r))_c = 2$,
  $(\pi_3(r))_c = 3$.
\end{definition}
Finally, Cohn et al.~defined a strengthening of SUSP which requires
that every triple of rows witness the necessary overlap.
\begin{definition}[Local SUSP]
\label{def:local-strong-USP}
A local strong uniquely solvable puzzle is an \puz{s}{k} where for
each triple of rows $u, v, w \in P$ with $u, v, w$ not all equal, there
exists $c \in [k]$ such that $(u_c, v_c, w_c)$ is an element
of $$\mathcal{L} = \set{(1, 2, 1), (1, 2, 2), (1, 1, 3), (1, 3, 3),
  (2, 2, 3), (3, 2, 3)}.$$
\end{definition}
Every SUSP $P$ corresponds to a much larger local SUSP $P'$, which,
informally, is the result of concatenating and duplicating the rows of
$P$ to explicitly demonstrate the $\forall \pi_1, \pi_2, \pi_3$ part
of \autoref{def:strong-USP}.
\begin{proposition}[{\cite[Proposition 6.3]{cksu05}}]
  \label{prop:susp-to-local}
  Let $P$ be a \susp{s}{k}, then there is a local \susp{s!}{s \cdot k} $P'$.
\end{proposition}
Note that in all of the definitions, local, strong, uniquely solvability
is invariant to the ordering of the rows of the puzzle, because $P$ is
a set---we use this fact implicitly.

Cohn et al.~show the following connection between the existence of
strong USPs and upper bounds on the exponent of matrix multiplication
$\omega$.
\begin{lemma}[{\cite[Corollary 3.6]{cksu05}}]
  \label{lem:omega}
  Let $\epsilon > 0$, if there is a strong uniquely solvable
  $(s,k)$-puzzle, there is an algorithm for multiplying $n$-by-$n$
  matrices in time $O(n^{\omega+\epsilon})$ where
  $$\omega \le \min_{m \in \mathbb{N}_{\ge 3}} \left( \frac{3 \log
    m}{\log(m-1)} - \frac{3 \log s!}{s\cdot k \log(m-1)}\right).$$
\end{lemma}
\noindent This result motivates the search for large strong USPs that
would result in faster algorithms for matrix multiplication.  In the
same article, the authors also demonstrate the existence of an
infinite family of strong uniquely solvable puzzles, for width $k$
divisible by three, that achieves a non-trivial bound on $\omega$.
\begin{lemma}[{\cite[Proposition 3.8]{cksu05}}]
  \label{lem:family-exists}
  There is an infinite family of strong uniquely solvable puzzles that
  achieves $\omega < 2.48$.
\end{lemma}
Finally, they conjecture that strong uniquely solvable
puzzles provide a route to achieving quadratic-time matrix
multiplication.
%% \begin{conjecture}[{\cite{cksu05}}]
%%   There exists a family of strong uniquely solvable puzzles that
%%   implies $\omega = 2$.
%% \end{conjecture}
Unfortunately, as mentioned in the introduction, this
conjecture was shown to be false.
\begin{lemma}[\cite{bccgu16}]
  Strong uniquely solvable puzzles cannot show $\omega < 2 +
  \epsilon$, for some $\epsilon > 0$.
\end{lemma}
That said, there remains hope that the uniquely solvable
puzzle approach could beat the refinements of Coppersmith-Winograd
even if it cannot reach $\omega = 2$.

%% \noindent This result is a consequence of a recent breakthrough
%% arithmetic progressions in cap sets \cite{clp17} combined with a
%% conditional result on the Erd\"{o}s-Szemeredi sunflower conjecture
%% \cite{asu13}.  The results of \cite{bccgu16} do imply that Cohn and
%% Umans' strong uniquely solvable puzzle approach cannot achieve the
%% ideal $\omega = 2$.  However, we are unaware of a concrete lower bound
%% on $\epsilon$ implies by this result.  This means there is a still a
%% substantial gap in our understanding between what has been achieved by
%% the refinements of LeGall, Williams, and Stothers, and the
%% impossibility of showing $\omega = 2$ using the Cohn et al.'s
%% approach.

%% file: verify.tex
\section{Verifying Strong USPs}
\label{sec:verify}

The core focus of this article is the problem of verifying strong
USPs, i.e., given an \puz{s}{k} $P$, output YES if $P$ is a strong
USP, and NO otherwise.  In this section we discuss the design of
algorithms to solve this computational problem as a function of the
natural parameters $s$ and $k$.

%% Along the way we also discuss some aspects of our practical
%% implementation that informed or constrained our designs.

All of the exact algorithms we develop in this section have worst-case
exponential running time.  However, asymptotic worst-case running time
is not the metric we are truly interested in.  Rather we are
interested in the practical performance of our algorithms and their
capability for locating new large strong USPs.  The algorithm that we
ultimately develop is a hybrid of a number of simpler algorithms and
heuristics.

We begin by discussing a na\"ive brute force algorithm based on the
definition of strong USP (\autoref{subsec:brute-force}), see how it
motivations a reduction to the 3D matching problem
(\autoref{sec:3DM}), and then how we might formulate a reduction to
the satisfiability and integer programming problems (Subsections
\ref{subsec:sat} \& \ref{subsec:mip}).  We then describe several
verification heuristics based on properties of strong USP
(\autoref{sec:heuristic}) and combine them with the verification
algorithms to produce a hybrid algorithm \textsc{Verify}
(\autoref{subsec:hybrid}).  As we discuss in
\autoref{subsec:performance}, our hybrid algorithm is quickly able to
check whether a given puzzle is a strong USP and aid in the search for
strong USP.

\subsection{Brute Force}
\label{subsec:brute-force}

The obvious algorithm for verification comes directly from the
definition of a strong USP.  Informally, we consider all ways of
permuting the twos and threes pieces relative to the ones pieces and
check whether the non-overlapping condition of
\autoref{def:strong-USP} is met.  A formal description of the
algorithm is found in \autoref{alg:brute-force}.

\begin{algorithm}[t]
  \caption{: Brute Force Verification}
  \label{alg:brute-force}
\begin{algorithmic}[1]
  \Require{An $(s,k)$-puzzle $P$.}
  \Ensure{YES, if $P$ is a strong USP and NO otherwise.}
  \Function{VerifyBruteForce}{$P$}
  \For{$\pi_2 \in \Sym{P}$}
    \For{$\pi_3 \in \Sym{P}$}
      \If{$\pi_2 \neq 1 \vee \pi_3 \neq 1$}
        \State{$found = false.$}
        \For{$r \in P$}
          \For{$i \in [k]$}
          \If{$\delta_{r_i, 1} + \delta_{(\pi_2(r))_i, 2} +
            \delta_{(\pi_3(r))_i, 3} = 2$} $found = true$. \EndIf
          \EndFor
        \EndFor
        \If{\textbf{not} $found$} \Return{NO.} \EndIf
      \EndIf
    \EndFor
  \EndFor
  \State{\Return{YES}.}
  \EndFunction
\end{algorithmic}
\end{algorithm}

The ones in Line 4 of \autoref{alg:brute-force} denote the
identity in $\Sym{P}$, and $\delta_{a,b}$ is the Kronecker delta
function which is one if $a = b$ and zero otherwise.  Observe that
\autoref{alg:brute-force} does not refer to the $\pi_1$ of
\autoref{def:strong-USP}.  This is because the strong USP
property is invariant to permutations of the rows and so $\pi_1$ can
be thought of as an arbitrary phase.  Hence, we fix $\pi_1 = 1$ to
simplify the algorithm.  Seeing that $|\Sym{P}| = s!$, we conclude
that the algorithm runs in time $O((s!)^2 \cdot s \cdot k \cdot
\poly{s})$ where the last factor accounts for the operations on
permutations of $s$ elements.  The dominant term in the running time
is the contribution from iterating over all pairs of permutations.
Finally, notice that if $P$ is a strong USP, then the algorithm runs
in time $\Theta((s!)^2 \cdot s \cdot k \cdot \poly{s})$, and that if
$P$ is not a strong USP the algorithm terminates early.  The
algorithm's poor performance made it unusable in our implementation,
however, its simplicity and direct connection to the definition made
its implementation a valuable sanity check against later more
elaborate algorithms (and it served as effective onboarding to the
undergraduate students collaborating on this project).

Although \autoref{alg:brute-force} performs poorly, examining
the structure of a seemingly trivial optimization leads to
substantially more effective algorithms. Consider the following
function on triples of rows $a, b, c \in P$: $f(a,b,c) = \vee_{i \in
  [k]} (\delta_{a_i,0} + \delta_{b_i,1} + \delta_{c_i,2} = 2).$ We can
replace the innermost loop in Lines 7 \& 8 of
\autoref{alg:brute-force} with the statement $found = found \vee
f(r, \pi_1(r), \pi_2(r))$.  Observe that $f$ neither depends on $P$,
$r$, nor the permutations, and that \autoref{alg:brute-force} no
longer depends directly on $k$.  To slightly speed up
\autoref{alg:brute-force} we can precompute and cache $f$ before
the algorithm starts and then look up values as the algorithm runs.
We precompute $f$ specialized to the rows in the puzzle $P$, and call
it $f_P$.

%% There are two obvious options to consider we can either precompute
%% $f$ specialized to the rows in the puzzle $P$, or we can precompute
%% $f$ for all possible rows.  In former case the time to precompute
%% $f$ is $\Theta(s^3 \cdot k)$ and in the later case $\Theta(3^{3k}
%% \cdot k)$.  The storage requirements are $\Theta(s^3)$ and
%% $\Theta(3^{3k})$ bits respectively.  The former is problematic for
%% large $s$ and later problematic even for small $k$.  Moreover the
%% combined running time for the two options with a single call to
%% verify a puzzle is $\Theta(s^3 \cdot k + (s!)^2 \cdot \poly{s})$
%% and $\Theta(3^{3k} \cdot k + (s!)^2 \cdot \poly{s})$.  In the
%% former case there is an asymptotic improvement, but in the later
%% case, the saving of a factor of $k$ is easily offset by the
%% additional $3^{3k}$ term.  For this reason we rule out the later
%% option and chose to represent the function $f$ specialized to $f_P$
%% for a given puzzle $P$.

\subsection{Strong USP Verification to 3D Matching}
\label{sec:3DM}

It turns out to be more useful to work with $f_P$ than with $P$.  It
is convenient to think of $f_P$ as a function $f_P : P \times P \times
P \rightarrow \set{0, 1}$ that is the complement of the characteristic
function of the relations of a tripartite hypergraph $H_P = \langle P
\sqcup P \sqcup P, \bar{f_P}\rangle$ where the vertex set is the
disjoint union of three copies of $P$ and $f_P$ indicates the edges
that are not present in $H_P$.

Let $H = \langle P \sqcup P \sqcup P, E \sse P^3\rangle$ be a
tripartite 3-hypergraph.  We say $H$ has a \emph{3D matching}
($\mathrm{3DM}$) iff there exists a subset $M \sse E$ with $|M| = |P|$
and for all distinct edges $e_1, e_2 \in M$, $e_1$ and $e_2$ are
\emph{vertex disjoint}, i.e., $e_1 \cap e_2 = \emptyset$.  Determining
whether a hypergraph has a 3D matching is a well-known
$\mathrm{NP}$-complete problem (c.f., e.g., \cite{gj79}).  We say that
a 3D matching is \emph{non-trivial} if it is not the set
$\condset{(r,r,r)}{r \in P}$.  \autoref{fig:hypergraph}
demonstrates a 3-hypergraph with a non-trivial 3D matching.

The existence of non-trivial 3D matchings in $H_P$ is directly tied to
whether $P$ is a strong USP.

\begin{figure}[t]
  \centering
  \includegraphics{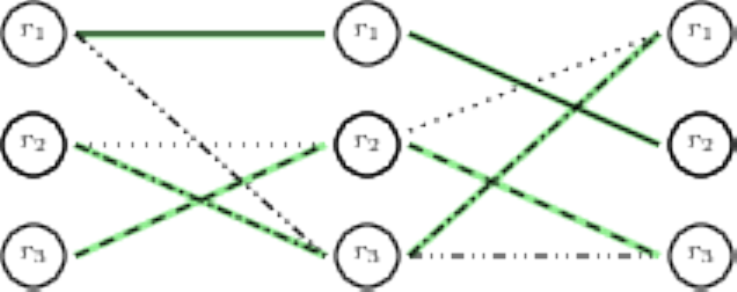}
  \caption{An example hypergraph $G$ with edges $E = \set{(r_1,
      r_1, r_2), (r_1, r_3, r_3), (r_2, r_2, r_1), \allowbreak
      (r_2, r_3, r_1), \allowbreak (r_3, r_2, r_3)}$.  The
    highlighted edges are a non-trivial 3D matching
    $M=\set{(r_1, r_1, r_2), (r_2, r_3, r_1),\allowbreak (r_3,
      r_2, r_3)}$ of $G$.}
  \label{fig:hypergraph}
\end{figure}

\begin{lemma}
  \label{lem:verify-to-3dm}
  A puzzle $P$ is a strong USP iff $H_P$ has no non-trivial 3D
  matching.
\end{lemma}

\begin{proof}
  We first argue the reverse.  Suppose that $H_p$ has a non-trivial 3D
  matching $M$.  We show that $P$ is not a strong USP by using $M$ to
  construct $\pi_1, \pi_2, \pi_3 \in \Sym{P}$ that witness this.  Let
  $\pi_1$ be the identity permutation.  For each $r \in P$, define
  $\pi_2(r) = q$ where $(r,q,*) \in M$.  Note that $q$ is well defined
  and unique because $M$ is 3D matching and so has vertex disjoint
  edges.  Similarly define $\pi_3(r) = q$ where $(r,*,q) \in M$.
  Observe that by construction $$M =
  \condset{(\pi_1(r),\pi_2(r),\pi_3(r))}{r \in P}.$$ Since $M$ is a
  matching of $H_P$, $M \sse \bar{f_P}$.  Because $M$ is a non-trivial
  matching at least one edge in $(a,b,c) \in M$ has either $a \neq b$,
  $a \neq c$, or $b \neq c$.  This implies, respectively, that as
  constructed $\pi_1 \neq \pi_2$, $\pi_1 \neq \pi_3$, or $\pi_2 \neq
  \pi_3$.  In each case we have determined that $\pi_1$, $\pi_2$, and
  $\pi_3$ are not all identical.  Thus we determined permutations such
  that for all $r \in P$, $f(\pi_1(r), \pi_2(r), \pi_3(r)) = 0$.  This
  violates Condition (ii) of \autoref{def:strong-USP}, hence
  $P$ is not a strong USP.

  The forward direction is symmetric.  Suppose that $P$ is not a
  strong USP. We show that $H_P$ has a 3D matching.  For $P$ not to be
  a strong USP there must exist $\pi_1, \pi_2, \pi_3 \in \Sym{P}$ not
  all identical such that Condition (ii) of
  \autoref{def:strong-USP} fails.  Define $e(r) =
  (\pi_1(r),\pi_2(r),\pi_3(r))$ and $M = \condset{e(r)}{r \in P}$.
  Since Condition (ii) fails, we have that $f_P(e(r)) = false$ for all
  $r \in P$.  This means that for all $r \in P$, $e(r) \in \bar{f_P}$
  and hence $M \sse \bar{f_P}$.  Since $\pi_1$ is a permutation, $|M|
  = |P|$.  Observe that $M$ is non-trivial because not all of the
  permutations are identical and there must be some $r \in P$ with
  $e(r)$ having non-identical coordinates.  Thus $M$ is a non-trivial
  3D matching.\qed
\end{proof}

As a consequence of \autoref{def:strong-USP}, strong-USP verification
is in \coNP{}.  Note that although 3D matching is an \NP-complete
problem, \autoref{lem:verify-to-3dm} does not immediately imply that
verification of strong USPs is \coNP-complete because $H_P$ is not an
arbitrary hypergraph.  It remains open whether strong-USP verification
is \coNP-complete.  \autoref{lem:verify-to-3dm} implies that to verify
$P$ is a strong USP it suffices to determine whether $H_P$ has a
non-trivial 3D matching.  In the subsequent subsections we examine
algorithms for the later problem.  We can, in retrospect, view
\autoref{alg:brute-force} as an algorithm for solving 3D matching.

We note that the parameters $s$ and $k$ are not fully independent.
First, $s \le 3^k$ because the maximum number of rows in a puzzle of
width $k$ is $|[3]^k| = 3^k$.  Second, we eliminate the dependence on
$k$ entirely by transforming an \puz{s}{k} into a 3D matching instance
on the vertex set $[s]^3$.  However, this transformation is not
without cost, because the size of $H_P$ is a function of the cube of
$s$ rather than linear in the size of the puzzle $s \cdot k$.

\subsection{Dynamic Programming}
\label{sec:dp}

The realization that the verification of strong USPs is a
specialization of 3D matching leads to a dynamic programming algorithm
for verification that runs in linear-exponential time $O(2^{2s}
\poly{s} + \poly{s,k})$.  The reduction allows us to replace the
permutations from $\Sym{P}$ with subsets of $P$ and effectively reduce
the cost of the outer loops of \autoref{alg:brute-force} from $s! =
\Theta(2^{s\log s})$ to $2^s$.

\autoref{alg:bi} describes a recursive bidirectional dynamic
programming algorithm for strong-USP verification that uses the 3D
matching instance.
\begin{algorithm}[t]
  \caption{: Bidirectional Dynamic Programming Verification}
  \label{alg:bi}
\begin{algorithmic}[1]
  \Require{An $(s,k)$-puzzle $P$.}
  \Ensure{YES, if $P$ is a strong USP and NO otherwise.}
  \Function{VerifyDynamicProgramming}{$P$}
  \State{Let $T = \emptyset$.}
  \State{Construct 3D matching instance $H_P$.}
  \Function{SearchHalf}{$\ell, Q,\ell_Q, R,\ell_R, \delta, t$}
  \If{$\ell = t$}
  ~~\If{$\delta = 1$} \Comment{Forward Base Case}
  ~~~~\State{Insert $(Q,R)$ into $T$.}
  ~~~~\State{\Return{$false$}.}
  ~~\Else \Comment{Reverse Base Case}
  ~~~~\If{$(P-Q, P-R) \in T$} \State{\Return{$true$}.} \Else
  \State{\Return{$false$}.} \EndIf
  ~~\EndIf
  \EndIf
  \State{$res = false$.} \Comment{Recursive Case}
  \For{$\ell'_Q = \ell_Q + 1$ to $s$}
  ~~\For{$\ell'_R = \ell_R + 1$ to $s$}
  ~~~~\If{$(p_\ell, p_{\ell'_Q}, p_{\ell'_R}) \in H_P \wedge \neg res$}
  ~~~~~~\State{$res =$ \textsc{SearchHalf}$(\ell + \delta, Q \cup
    \set{p_{\ell'_Q}}, \ell'_Q, R \cup \set{p_{\ell'_R}}, \ell'_R,
    \delta, t)$.} \EndIf
    \EndFor
  \EndFor
  \State{\Return{$res$.}}
  \EndFunction
  \State{\textsc{SearchHalf}$(1,\emptyset, 0, \emptyset, 0, 1, \lfloor
    s / 2 \rfloor + 1)$.}
  \State{\Return{\textsc{SearchHalf}$(s, \emptyset, 0, \emptyset, 0,
      -1, \lfloor s/2 \rfloor)$}.}
  \EndFunction
\end{algorithmic}
\end{algorithm}
The algorithm consists of two phases. Let $t = \lfloor s/2
\rfloor$. The first phase determines all possible sets $Q,R \sse P$
with $|Q| = |R| = t$ such that there is 3D matching $M_1$ of $H_P$
when restricted to the vertices $\set{p_1,p_2, \ldots, p_t} \sqcup Q
\sqcup R$.  The sets $Q,R$ satisfying the requirement are stored in a
table $T$ during the first phase on Line 7.  The second phase
determines all possible sets $Q,R \sse P$ with $|Q| = |R| = s - t$
such that there is a 3D matching $M_2$ of $H_P$ when restricted to the
vertices $\set{p_{t+1},p_{t+2},\ldots,p_s} \sqcup Q \sqcup R$.  For
each pair $(Q,R)$ the algorithm considers in the second phase, it
checks whether $(P - Q, P - R)$ was inserted into $T$ during the first
phase.  If the pair is present, it means that there is a 3D matching
of $H_P$ which is $M = M_1 \cup M_2$.  This works because, by Line 10,
$M_1$ and $M_2$ are partial 3D matchings on $\set{p_1,\ldots,p_t}
\sqcup (P - R) \sqcup (P - Q)$ and $\set{p_{t+1},\ldots p_s} \sqcup R
\sqcup Q$, respectively, which implies that $M_1$ and $M_2$ are vertex
disjoint.  The first phase always returns $false$, which is ignored,
and the second phase returns whether a complete matching could be
found, and, hence, by \autoref{lem:verify-to-3dm}, whether $P$ is a
strong USP.

The running time of this algorithm is dominated by the number of pairs
of sets $(Q,R)$ it examines.  Observe that rows of $P$ are considered
in order in Lines 15 \& 16.  Further, the algorithm tracks the index
of the last elements added to $Q$ and $R$ in $\ell_Q$ and $\ell_R$,
respectively.  The algorithm only adds new elements to $Q$ or $R$ that
have higher indexes than ones previously added.  Altogether this
implies that each pair of sets $(Q,R)$ is only considered at most once
during a phase.  Since $Q, R \sse P$, there are at most $\sum_{i =
  0}^t \binom{s}{i} \cdot \binom{s}{i} \le (\sum_{i = 0}^t
\binom{s}{i})^2 \le (2^s)^2 = 4^s$ pairs $(Q,R)$.  This means that
\textsc{SearchHalf} is called at most $4^s$ times during each phase.
Hence the running time of the algorithm is $O(4^s \cdot s^2 \cdot
\poly{s} + T_{3DM}(s,k))$ where $s^2$ factor comes from the inner
loops, $\poly{s}$ time to manipulate the sets and track the contents
of $T$ as a hash table, and $T_{3DM}(s,k)$ accounts for the time to
construct $H_P$.  The memory requirements of \autoref{alg:bi} are
similarly high---the first phase uses $O(4^s \cdot s)$ bits to store
$T$.

Note that \autoref{alg:bi} does not early terminate on $P$ that are
strong USP, because it must search through all pairs before
determining that none can be found.  The algorithm could be modified
to allow early termination when $P$ is not a strong USP by causing the
second phase of search to immediately return in Line 18 once the first
3D matching witness has been located.  However, this still requires
the first phase to run to completion.  A remedy for this would be to
run both phases in parallel and have them check against each other.
We chose not to because it would substantially complicate the
implementation and would be unlikely to ultimately improve the
performance of our combined algorithms.

%% --- Reduce discussion of implementation details, and move to
%%   implementation section?---
%% In the implementation of \autoref{alg:bi} we represented the sets
%% $Q,R$ using bit sets encoded into a single 64-bit \texttt{long} which
%% permitted it to represent these sets for $s \le 32$.  We implemented
%% $T$ using a \texttt{map}s from the C++ standard template library.
%% These choices made the basic data structure operations in
%% implementation effectively constant time, and have low memory overhead
%% beyond the contents of $T$.

For comparison, more advanced techniques like those of Bj\"{o}rklund
et al.~can achieve a better asymptotic time of $O(2^s\poly{s})$
\cite{bhkk17}.  We chose not to implement their algorithm, because we
judged that it would not substantially increase the domain for which
verification was possible.

\subsection{3D Matching to Satisfiability} 
\label{subsec:sat}

By \autoref{lem:verify-to-3dm}, one can determine whether a puzzle
$P$ is a strong USP by constructing the graph $H_P$ and deciding
whether it has a non-trivial 3D matching.  Here we reduce our 3D
matching problem to the satisfiability (SAT) problem on conjunctive
normal form (CNF) formulas and then use a state-of-the-art SAT solver
to resolve the reduced problem. To perform the reduction, we convert
the graph $H_P$ into a CNF formula $\Psi_P$, a depth-2 formula that is
the AND of ORs of Boolean literals.  We construct $\Psi_P$ so that
$\Psi_P$ is satisfiable iff $H_P$ has a non-trivial 3D matching.

Let $H_P = \langle V = P \sqcup P \sqcup P, E \subseteq P^3 \rangle$
be the 3D matching instance associated with the puzzle $P$.  Our goal
is to determine whether there is a non-trivial 3D matching $M
\subseteq E$.  A na\"{i}ve reduction would be to have variables $M_{u,v,w}$
indicating inclusion of each edge $(u,v,w) \in P^3$ in the
matching. This results in a formula $\Psi_P$ with $s^3$ variables and
size $\Theta(s^5)$ because including an edge $e \in P^3$ excludes the
$\Theta(s^2)$ edges $e'$ with $e \cap e' \neq \emptyset$.  To decrease
the size of $\Psi_P$ we instead use sets of variables to indicate
which vertices in the second and third part of $V$ are matched with
each vertex in the first part.  In particular we have Boolean
variables $M_{u,v}^1$ and $M_{u,w}^2$ for all $u, v, w \in P$, and
these variable map to assignments in the na\"{i}ve scheme in the
following way: $M_{u,v}^1 \wedge M_{u,w}^2 \Leftrightarrow M_{u,v,w}$.

We now write our CNF formula for 3D matching.  First, we have clauses
that prevents non-edges from being in the matching:
\begin{equation}
  \Psi_P^{\textrm{non-edge}} = \bigwedge_{(u,v,w) \in \overline{E}}
  (\neg M_{u,v}^1 \vee \neg M_{u,w}^2).
\end{equation}
Second, we add clauses require that every vertex in $H_P$ is matched
with some edge:
\begin{equation}
\begin{split}
  \Psi_P^{\ge 1} =
  &\left(\bigwedge_{u \in P} (\vee_{v\in P} ~M_{u,v}^1) \wedge (\vee_{w \in P} ~M_{u,w}^2)\right) \\
  &\wedge \left(\bigwedge_{v \in P} (\vee_{u \in P} ~M_{u,v}^1)\right)
  \wedge \left(\bigwedge_{w \in P} (\vee_{u \in P} ~M_{u,w}^2)\right).
  \end{split}	
\end{equation}
%% Note that these clauses are no longer symmetric because the first part
%% of $H_P$ is considered in both sets of variables.
Third, we require that each vertex be matched with at most one edge
and so have clauses that exclude matching edges that overlap on one or
two coordinates.
\begin{equation}
  \Psi_P^{\le 1} = \bigwedge_{i \in \set{1,2}} \bigwedge_{(u,v),
    (u',v') \in P^2} (u = u' \vee v = v') \wedge (u,v \neq u',v')
  \Rightarrow \neg M_{u,v}^i \vee \neg M_{u',v'}^i.
\end{equation}
%% Observe that $\Psi_P^{\le 1}$ has size $\Theta(s^3)$ compared to the
%% size of the original $\Phi_P^{\le 1}$ that had size $\Theta(s^5)$.
Fourth, we exclude the trivial 3D matching by requiring that at least
one of the diagonal edges not be used:
%\begin{equation}
  $\Psi_P^{\textrm{non-trivial}} = \bigvee_{u \in P} \neg M_{u,u}^1 \vee
  \neg M_{u,u}^2.$
%\end{equation}
Finally, we AND these into the overall CNF formula:
%\begin{equation}
  $\Psi_P = \Psi_P^{\textrm{non-edge}} \wedge \Psi_P^{\le 1} \wedge
  \Psi_P^{\ge 1} \wedge \Psi_P^{\textrm{non-trivial}}.$
%\end{equation}
The size of the CNF formula $\Psi_P$ is $\Theta(s^3)$, has $2s^2$
variables, and is a factor of $s^2$ smaller than the na\"{i}ve
approach.  Thus we reduce 3D matching to satisfiability by converting
the instance $H_P$ into the CNF formula $\Psi_P$.

%% To solve the reduced satisfiability instance we used the open-source
%% solver MapleCOMPSPS from the 2016 International SAT Competition
%% \cite{bhj17}.  This solver is conflict driven and uses a learning rate
%% branching heuristic to decide which variables are likely to lead to
%% conflict and has had demonstrable success in practice \cite{lgpc16}.
%% We chose MapleCOMPSPS because it was state of the art at the time our
%% project started.  It is likely that more recently-developed solvers
%% would achieve similar or better performance on our task.

%% \autoref{alg:sat} formalizes the reduction from strong USP
%% verification to SAT.

%% \begin{algorithm}[h]
%%   \caption{: Reduction to satisfiability}
%%   \label{alg:sat}
%% \begin{algorithmic}[1]
%%   \Require{An $(s,k)$-puzzle $P$.}
%%   \Ensure{YES, if $P$ is a strong USP and NO otherwise.}
%%   \Function{VerifySAT}{$P$}
%%   \State{Construct 3D matching instance $H_P$.}
%%   \State{Construct CNF formula $\Psi_P$.}
%%   \If{$\Psi_P$ is satisfiable}
%%   \State{\Return NO}
%%   \Else
%%   \State{\Return YES}
%%   \EndIf
%%   \EndFunction
%% \end{algorithmic}
%% \end{algorithm}

\subsection{3D Matching to Integer Programming}
\label{subsec:mip}

In parallel to the previous subsection, we use the connection between
verification of strong USPs and 3D matching to reduce the former to
integer programming, another well-known \NP{}-complete problem (c.f.,
e.g., \cite{kv12}) and then apply a state-of-the-art solver to resolve
it.  Again, let $H_P = \langle V, E \rangle$ be the 3D matching
instance associated with $P$.  We construct an integer program $Q_P$
over $\set{0,1}$ that is infeasible iff $P$ is a strong USP.  Here the
reduction is simpler than the previous one because linear constraints
naturally capture matching.

We use $M_{u,v,w}$ to denote a variable with values in $\set{0,1}$ to
indicate whether the edge $(u,v,w) \in P^3$ is present in the
matching. To ensure that $M$ is a subset of $E$ we add the following
edge constraints to $Q_P$: $\forall u,v,w \in P, \forall (u, v, w)
\not\in E, M_{u,v,w} = 0.$ We also require that each vertex in each of
the three parts of the graph is incident to exactly one edge in $M$.
This is captured by the following vertex constraints in $Q_P$:
$\forall w \in P, \sum_{u,v \in P} M_{u,v,w} = \sum_{u,v \in P}
M_{u,w,v} = \sum_{u,v \in P} M_{w,u,v} = 1.$ Lastly, since we need
that the 3D matching be non-trivial we add the constraint: $\sum_{u
  \in P} M_{u,u,u} < |P|.$

To check whether $P$ is a strong USP we determine whether $Q_P$ is not
feasible, i.e., that no assignment to the variables $M$ satisfy all
constraints.  We note that reduction from 3D matching to IP is
polynomial time and that there are $s^3$ variables in $Q_P$, and that
the total size of the constraints is $s^3 \cdot \Theta(1) + 3s \cdot
\Theta(s^2) + 1 \cdot \Theta(s^3) = \Theta(s^3)$, similar to size of
$\Psi_P$ in the SAT reduction.
 
%% \autoref{alg:ip} formalizes the reduction from strong USP
%% verification to IP.

%% \begin{algorithm}[h]
%%   \caption{: Reduction to integer programming}
%%   \label{alg:ip}
%% \begin{algorithmic}[1]
%%   \Require{An $(s,k)$-puzzle $P$.}
%%   \Ensure{YES, if $P$ is a strong USP and NO otherwise.}
%%   \Function{VerifyIP}{$P$}
%%   \State{Construct 3D matching instance $H_P$.}
%%   \State{Create empty integer program $Q_P$ with Boolean variables
%%     $\condset{M_{u,v,w}}{u,v,w \in P}$.}
%%   \State{Add constraints of Section 3.5 to $Q_P$.}
%%   \If{\textsc{IsFeasible}$(Q_P)$}
%%   \State{\Return{NO}.}
%%   \Else
%%   \State{\Return{YES}.}
%%   \EndIf
%%   \EndFunction
%% \end{algorithmic}
%% \end{algorithm}

\subsection{Heuristics}
\label{sec:heuristic}

Although the exact algorithms presented in the previous sections make
substantial improvements over the brute force approach, the resulting
performance remains impractical.  To resolve this, we also develop
several fast verification heuristics that may produce the
non-definitive answer MAYBE in place of YES or NO.  Then, to verify a
puzzle $P$ we run this battery of fast heuristics and return early if
any of the heuristics produce a definitive YES or NO.  When all of the
heuristics result in MAYBE, we then run one of the slower exact
algorithms that were previously discussed.  The heuristics have
different forms, but all rely on the structural properties of strong
uniquely solvable puzzles.

\subsubsection{Downward Closure}
The simplest heuristics we consider is based on the fact that strong
USPs are downward closed.

\begin{lemma}
  \label{lem:downward-closed}
  If $P$ is a strong USP, then so is every subpuzzle $P' \sse P$.  
\end{lemma}

\begin{proof}
Let $P$ be a strong USP and $P' \sse P$.  By \autoref{def:strong-USP},
for every $(\pi_1, \pi_2, \pi_3) \in \Sym{P}^3$ not all identity,
there exist $r \in P$ and $i \in [k]$ such that exactly two of the
following hold: $(\pi_1(r))_i = 1$, $(\pi_2(r))_i = 2$, $(\pi_3(r))_i
= 3$.  Consider restricting the permutations to those that fix the
elements of $P \backslash P'$.  For these permutations it must be the
case that $r \in P'$ because otherwise $r \in P \backslash P'$ and
there is exactly one $j \in [3]$ for which $(\pi_j(r))_i = j$ holds.
Thus we can drop the elements of $P \backslash P'$ and conclude that
for every tuple of permutations in $\Sym{P'}$ the conditions of
\autoref{def:strong-USP} hold for $P'$, and hence that $P'$ is a
strong USP. \qed
\end{proof}

This leads to a polynomial-time heuristic that can determine that a
puzzle is not a strong USP.  Informally, the algorithm takes an
$(s,k)$-puzzle $P$ and $s' \le s$, and verifies that all subsets $P'
\subseteq P$ with size $|P'| = s'$ are strong USPs.  If any subset
$P'$ is not a strong USP, the heuristic returns NO, and otherwise it
returns MAYBE.  For completeness, this algorithm is described in
\autoref{alg:heuristic-downward-closed}.

\begin{algorithm}[t]
  \caption{: Downward-Closure Heuristic}
  \label{alg:heuristic-downward-closed}
  \begin{algorithmic}[1]
    \Require{An $(s,k)$-puzzle $P$, and size $s' \le s$.}
    \Ensure{NO, if $P$ has a set of $s'$ rows that do not form a strong
      USP, and MAYBE otherwise.}
    \Function{HeuristicDownwardClosed}{$P, s'$}
    \For{$P' \sse P, |P'| = s'$}
    \If{$P'$ is not a strong USP} \Return{NO.} \EndIf
    \EndFor{}
    \State{\Return{MAYBE}.}
    \EndFunction
  \end{algorithmic}
\end{algorithm}

This algorithm runs in time $O({s \choose s'}\cdot T(s', k))$ where
$T(s',k)$ is the runtime for verifying an $(s',k)$-puzzle.  In
practice we did not apply this heuristic for $s'$ larger than $3$.
When $s'$ is some constant $d$, the running time becomes $O(s^d \cdot
T(d, k)) = O(s^d k)$ using the brute force algorithm
(\autoref{alg:brute-force}) for verification of the puzzle $P'$.

\subsubsection{Unique Pieces}
Every strong uniquely solvable puzzle is a uniquely solvable puzzle.
A necessary condition for a puzzle to be a USP is that for each
element in $[3]$, the collection of subrows contains no
duplicates.

\begin{lemma}[{Implicit in \cite{cksu05}}]
  \label{lem:unique-pieces}
  If $P$ is a USP, then for all $e \in [3]$, and distinct
  rows $r_1, r_2 \in P$, there is a column $c \in [k]$ were one
  of the rows $r_1$ or $r_2$ has an $e$ and the other one does not.
\end{lemma}

\begin{proof}
Suppose, for the sake of contradiction, that this is not the case, and
distinct rows $r_1, r_2 \in P$ have $e$ in exactly the same columns
for some $e \in [3]$.  We show that $P$ is not a USP.  Choose $\pi_e =
(r_1 r_2)$, i.e., the permutations that transposes the subrows for $e$
in rows $r_1$ and $r_2$.  Choose the other two permutations for the
elements of $[3] \backslash \set{e}$ to be the identity.  Since the
permutations are not all the identity, the second half of
\autoref{def:USP} applies.  However, the puzzle that results from the
permutations is identical to $P$ and for all $c \in [k]$ and each row
$r \in P$ there exists exactly on $i \in [3]$ where $(\pi_i(r))_c =
i$.  Hence the definition of uniquely solvable is not satisfied and we
have a contradiction. \qed
\end{proof}
Note that the reverse direction of \autoref{lem:unique-pieces} does
not hold.  The puzzle in \autoref{fig:puzzle} is an example of this:
It is not uniquely solvable, but the subrows for each element are
distinct.

We can make \autoref{lem:unique-pieces} effective as via a linear-time
heuristic capable of ruling out puzzles that are not (strong) USPs.
Although straightforward, for completeness we formalize our approach
in \autoref{alg:unique-pieces}.
\begin{algorithm}[t]
  \begin{algorithmic}[1]
    \Require{An \puz{s}{k} $P$.}
    \Ensure{NO, if a witness is found for $P$ not being a (strong)
      USP, and MAYBE otherwise.}
    \Function{HeuristicUniquePieces}{$P$}
    ~~\State{Initialize empty sets $S_1$, $S_2$, $S_3$.}
    ~~\For{$r \in P$}
    ~~~~\For{$e \in [3]$}
    ~~~~~~\State{Let $h = \condset{c \in [k]}{r_c = e}$.}
    ~~~~~~\If{$h \in S_e$} \Return{NO}. \EndIf
    ~~~~~~\State{$S_e = S_e \cup \set{h}$.}
    ~~~~\EndFor
    ~~\EndFor
    ~~\State{\Return{MAYBE}.}
    \EndFunction
  \end{algorithmic}
  \caption{: Unique Pieces Heuristic
  \label{alg:unique-pieces}}
  
\end{algorithm}
When the sets are implemented as hash tables, the expected running
time of this algorithm is $O(s \cdot k)$ time, which is linear in the
size of the puzzle $P$.  An alternative worst-case $O(s \cdot k)$ time
implementation uses radix sort to sort the characteristic sequences of
the subrows as binary numbers and then scans adjacent rows to to
detect duplication.

The unique pieces heuristic is equivalent to the
downward-closure heuristic for subpuzzles of size two.

\begin{lemma}
  \label{lem:heuristic-equiv}
Let $P$ be an \puz{s}{k}, then $\textsc{HeuristicUniquePieces}(P) =
\textsc{HeuristicDownwardClosed}(P, 2)$.
\end{lemma}

\begin{proof}
  We show both directions.

  Suppose that $P$ fails the unique pieces heuristic for, w.l.o.g.,
  $e=1$, then there are distinct rows $r_1, r_2 \in P$ where the cells
  that contain $1$ are all in the same columns.  This means we can
  swap those $1$'s subrows without causing overlap or changing the
  puzzle.  This implies that $P' = \set{r_1, r_2}$ is not a (strong)
  USP.  Since $|P'| = 2$ and $P' \sse P$, the downward closure
  heuristic for $s' = 2$ will also conclude that $P$ is not a (strong)
  USP.

  Suppose that $P$ fails the downward-closure heuristic for $s' = 2$.
  Then there is a pair of distinct rows $r_1, r_2 \in P$ for which $P'
  = \set{r_1, r_2}$ is not a strong USP.  Suppose there is no columns
  were $r_1$ and $r_2$ differ, then the subrows of $r_1$, $r_2$ are
  the same for all elements, and so $P$ fails the unique pieces
  heuristic.  For the other case, suppose there is a least one column
  $c \in [k]$ where $r_1$ and $r_2$ differ.  W.l.o.g., let that column
  be $((r_1)_c, (r_2)_c) = (1, 2)$.  Because $P'$ is not an SUSP and
  this column is $(1, 2)$, there can be other no columns that are in
  from the set $\set{(1, 3), (2, 3), (3, 2), (3, 1)}$ otherwise they
  would form an SUSP with the column $(1, 2)$.  This means the only
  columns that $P'$ contains are from the set $\set{(1, 2), (2, 1),
    (1, 1), (2, 2), (3, 3)}$.  Therefore, the columns which contain
  $2$ must match and the subrows for $2$ in $r_1$ and $r_2$ are
  identical.  Thus, $P'$, and so $P$, fails the unique pieces
  heuristic. \qed
\end{proof}
A corollary of this proof is that for size-two puzzles, every USP is
also a strong USP.
\begin{corollary}
  \label{cor:size-2}
  Let $P$ be a \puz{2}{k}, if $P$ is a uniquely solvable puzzle, then
  $P$ is a strong uniquely solvable puzzle.
\end{corollary}
Since the unique pieces heuristic is equivalent to the downward-closure
heuristic for $s'=2$ and the running time of unique pieces is linear
in the puzzle size, $O(s \cdot k)$, and the running time of downward
closed is $O(s^2 \cdot k)$, we use the unique pieces heuristic in
place of downward closed for $s' = 2$.

\subsubsection{Greedy}
This heuristic attempts take advantage of \autoref{lem:verify-to-3dm}
and greedily search for a 3D matching for the instance $H_P$.  The
heuristic proceeds iteratively, determining the vertex of the first
part of the 3D matching instance with the least edges and randomly
selecting an edge of that vertex to put into the 3D matching.  If the
heuristic successfully constructs a 3D matching it returns NO
indicating that the input puzzle $P$ is not a strong USP.  If the
heuristic reaches a point were prior commitments have made the
matching infeasible, the heuristic starts again from scratch.  This
process is repeated some number of times before it gives up and
returns MAYBE.  In our implementation we use $s^2$ attempts because it
is similar to the running time of the reductions and it empirically
reduced the number of instances requiring full verification in the
domain of puzzles with $k = 6, 7, 8$ while not increasing the running
time by too much.  The greedy heuristic is formalized in
\autoref{alg:random-greedy}.

\begin{algorithm}[t]
  \caption{: Greedy Heuristic}
  \label{alg:random-greedy}
  \begin{algorithmic}[1]
  \Require{An $(s,k)$-puzzle $P$, and iteration bound $t$.}
  \Ensure{NO, if a witness is found for $P$ not being a strong USP,
    and MAYBE otherwise.}
  \Function{HeuristicGreedy}{$P$}
  \State{Construct 3D matching instance $H_P$.}
  \For{$i = 1$ to $t$}
    \For{$u \in P$}
      \State{$cts[u] = \sum_{v,w \in P} H_P(u,v,w)$.} \Comment{Number
        of edges incident vertex $u$.} 
    \EndFor
    \State{Let $U,V,W = \emptyset.$}
    \State{Let $m = 0.$} \Comment{Number of edges in matching.}
    \While{$m < s$} 
    \State{Select $u \in \condset{w \in \bar{U}}{cts[w] = \max_{v \in
          \bar{U}} cts[v]}$ uniformly at random.}
      \If{$cts[u] = 0$} break. \EndIf
      \State{Let $D = \condset{(v,w) \in \bar{V} \times
          \bar{W}}{H_P(u,v,w) = 1}$.}
      \State{Select $(v,w) \in D$ uniformly at random.}
      \For{$v' \in P$} \Comment{Update edge counts.}
        \For{$w' \in P$}
          \If{$(v',w') \in \bar{V} \times \bar{W}$ and $H_P(u,v',w') =
            1$}
          \State{$cts[u]\texttt{--}$.} \EndIf
          \If{$(v',w') \in \bar{U} \times \bar{W}$ and $H_P(v',v,w') =
            1$ and $v' \neq u$} \State{$cts[v']\texttt{--}$.} \EndIf
          \If{$(v',w') \in \bar{U} \times \bar{V}$ and $H_P(v',w',w) =
            1$ and $v' \not\in \set{u, v}$}% \neq u$ and $v' \neq v$}
          \State{$cts[v']\texttt{--}$.} \EndIf
        \EndFor
      \EndFor
      \State{$U, V, W = U \cup \set{u}, V \cup \set{v}, W \cup
        \set{w}$.} \Comment{Add edge to matching.}
      \State{$m = m + 1.$}
      \EndWhile
      
    \If {$m \ge s$} \Return{NO}. \Comment{3D matching found so not
      SUSP, halt.} \EndIf
  \EndFor
  \State{\Return{MAYBE}.}
  \EndFunction
\end{algorithmic}
\end{algorithm}

The array $cts$ is used to store the number of edges $cts[u]$ that
remain associated with vertex $u$ along the first coordinate.  Much of
the algorithm is devoted to maintaining this invariant.  The sets
$U,V,W$ store the vertices along the three coordinates, respectively,
that have already been incorporated into the partial 3D matching.
Like in \autoref{alg:bi} we do not store the matching
itself, only the vertices involved.  The break at Line 10 triggers
when the partial 3D matching is a dead end and cannot be extended into
a full 3D matching.  The condition of Line 23 is true when a full 3D
matching has been constructed and causes the algorithm to return that
$P$ is not a strong USP.

The running time of this algorithm is $O(s^3 t + T_{3DM}(s,k))$, where
$T_{3DM}(s,k)$ is the time required to construct 3D matching instances
from $(s,k)$-puzzles.  This algorithm has the potential to be
considerably slower than the downward-closure heuristic, and in
practice we set $t = s^2$.  However, the main loop can terminate early
at Line 10 when it fails to extend the 3D matching, this permits the
expected time to much less than the worst case.  For a puzzle $P$ that
is a strong USP, the heuristic takes the full $\Omega(s^3 t +
T_{3DM}(s,k))$ time.

Compared to the downward-closure and unique pieces heuristics this
heuristic is much less efficient.  As a result we only run it when
when the other heuristics have failed.  See
\autoref{subsec:performance} for a comparison of effectiveness these
heuristics in our experiments.

\subsection{Hybrid Algorithm}
\label{subsec:hybrid}

Our final verification algorithm (\autoref{alg:hybrid}) is a hybrid of
several exact algorithms and heuristics.  The size thresholds for
which algorithm and heuristic to apply were determined experimentally
for small $k$ and are focused on the values where our strong USP
search algorithms are tractable $k \le 6$ (or nearly tractable $k \le
8$).  We decide to run both of the reductions to SAT and IP in
parallel because it is not clear which algorithm performs better in
general. Since verification halts when either algorithm completes, the
wasted effort is within a factor of two of what the better algorithm
could have done alone.  We also chose to do this because we
experimentally observed that there were many instances that one of the
algorithms struggled with that the other did not---this resulted in a
hybrid algorithm that out performed the individual exact algorithms on
average. We show in \autoref{subsec:performance} that our hybrid
algorithm and heuristics perform well in practice at quickly verifying
strong USPs for small width $k$.  Further, \autoref{sec:choice-sat}
contains a discussion of the relative performance of the SAT and IP
approaches on different instance types from our benchmark experiments.

\begin{algorithm}[t]
  \caption{: Hybrid Verification}
  \label{alg:hybrid}
\begin{algorithmic}[1]
  \Require{An $(s,k)$-puzzle $P$.}
  \Ensure{YES, if $P$ is a strong USP, and NO otherwise.}
  \Function{Verify}{$P$}
  \If{$s \le 2$} \Return{\Call{VerifyBruteForce}{$P$}.} \EndIf
  \State{Return result if \Call{HeuristicUniquePieces}{$P$} is not MAYBE.}
  \If{$s \le 7$} \Return{\Call{VerifyDynamicProgramming}{$P$}.} \EndIf
  \State{Return result if \Call{HeuristicDownwardClosed}{$P, 3$} is
    not MAYBE.}  
  \State{Return result if \Call{HeuristicGreedy}{$P$} is not MAYBE.}
  \State{Run \Call{VerifySAT}{$P$} and \Call{VerifyIP}{$P$} in
    parallel and return first result.}
  \EndFunction
\end{algorithmic}
\end{algorithm}

%% file: search.tex
\section{Searching for Strong USPs}
\label{sec:search}

With a practical verification algorithm in hand, we consider the
problem of searching for large strong USPs.  Because the set of strong
USPs is downward closed, a natural search strategy is: Start with the
empty set and repeatedly consider adding rows while maintaining the
strong-USP property.  However, while this strategy will lead to a
maximal-size strong USP, it is not guaranteed to produce a
maximum-size strong USP.  This is because the set of strong USPs does
not form a matroid, rather it is only an independence system (c.f.,
e.g., \cite{oxl06}).

In particular, (i) the empty puzzle is a strong USP and (ii) the set
of strong USP are downward closed by \autoref{lem:downward-closed}.
The final property required to be a matroid, the augmentation
property, requires that for every pair of strong USPs $P_1, P_2$ with
$|P_1| \le |P_2|$ there is a row of $r \in P_2 \backslash P_1$ such
that $P_1 \cup \set{r}$ is also a strong USP.  For a simple
counterexample consider the strong USPs $P_1 = \set{32}$ and $P_2 =
\set{12, 23}$.  Using \autoref{lem:unique-pieces}, we see that neither
$P_1 \cup \set{12} = \set{12, 32}$ nor $P_1 \cup \set{23} = \set{23,
  32}$ are strong USPs, and hence the augmentation property fails.
One consequence is that na\"{i}ve greedy algorithms will likely be
ineffective for finding maximum-size strong USPs.  Furthermore, we do
not currently know of an efficient algorithm that can take a strong
USP $P$ and determine a row $r$ such that $P \cup \set{r}$ is a strong
USP.

Despite that, we have had some success in applying general-purpose
tree-search techniques with pruning based on the symmetries of strong
USPs together with our practical verification algorithm to construct
maximum-size strong USPs for small $k$.

\subsection{Puzzle Symmetry}
\label{subsec:symmetry}

Since puzzles are defined as sets of rows, the ordering of the rows of
a puzzle $P$ does not affect the SUSP property.  Similarly, but
slightly less obviously, the SUSP property is invariant to reordering
the columns of the puzzle, because the required existential condition
$\exists c \in [k] \mathrm{~st.~} (...)$ from \autoref{def:strong-USP}
is independent of the ordering of the columns.  Lastly, the alphabet
$[3]$ typically used to represent the elements of a puzzle is
completely arbitrary, any set of three distinct values would suffice.
These values are not interpreted mathematically, aside from their
convenience in expressing the SUSP definition concisely.  This logic
can be formalized into the following lemma.
\begin{lemma}
  \label{lem:SUSP:invariance}
  Let $\rho \in \Sym{[k]}, \delta \in \Sym{[3]}$.  A \puz{s}{k}
  $P$ is a strong USP iff $ \condset{(\delta(r_{\rho(c)}))_{c \in
      [k]}}{r \in P}$ is a strong USP.
\end{lemma}
\begin{proof} Follows immediately from \autoref{def:puzzle} and
  \autoref{def:strong-USP}. \qed
\end{proof}
This lemma implies that the SUSP property is invariant with respect to
these kinds of puzzle transformations.  We call two puzzles $P, P'$
that are related in this way \emph{isomorphic}, and use the notation
$P \isom P'$ to denote this.  The relation $\isom$ is an equivalence
relation, because permutations are invertable, and so it partitions
the set of puzzles into equivalence classes.

This notion of isomorphism is naturally related to the same notion in
graphs. For each \puz{s}{k} $P$ we can define a colored, undirected graph
$G_P$.  This graph consists of vertices that are partitioned into
four sets of different colors: $V = \set{row_r}_{r \in [s]} \sqcup
\set{col_{c}}_{c \in [k]} \sqcup \set{e_i}_{i \in [3]} \sqcup
\set{v_{r, c}}_{(r, c) \in [s] \times [k]}$.  There are $s + k + 3 + s
\cdot k$ vertices in $G_P$.  The first three parts are vertices
representing the rows and columns of $P$, and the elements of
$[3]$, respectively, and the fourth part are vertices for each
of the $s \cdot k$ cells in the $P$.  The edge relation of $G_P$ is
straightforward: Each vertex $v_{r, c}$ is connected to three vertices
corresponding to the row, columns and element that the cell indexed
$(r,c)$ contains in $P$.  In particular, the three edges attached to
$v_{r,c}$ are $(v_{r,c}, row_r), (v_{r,c}, col_c), (v_{r,c},
elt_{P(r,c)})$.  In total, $G_P$ has $3 \cdot s \cdot k$ edges.
Because the vertex sets for rows, columns, and elements are each
uniquely colored and each cell of $P$ is connected to vertices
representing its row, column, and element, the automorphisms of $G_P$
are in 1-1 correspondence to the automorphisms of $P$ under
permutations of rows, columns, and elements.  This implies that for two
\puzs{s}{k} $P, P'$, if $G_P \isom G_{P'}$ then there exists
permutations of the rows, columns, and elements of $P$ which results
in $P'$.  Further by \autoref{lem:SUSP:invariance}, if $G_P \isom
G_{P'}$, then $P \isom P'$, and $P$ is an SUSP iff $P'$ is an SUSP.

\subsection{Symmetry-Pruned Tree Search}
\label{subsec:tree-search}

A natural way to search for strong USPs is based on breadth-first
search and uses the fact that strong USP are downward closed
(\autoref{lem:downward-closed}): To find the largest possible
width-$k$ strong USP, (i) start with all possible first rows -- the $3^k$
\puzs{1}{k}, (ii) attempt to extend the resulting puzzles with all
possible rows keeping only the puzzles that are strong USPs and which are not
isomorphic to the strong USPs that have been seen before to form the new search
frontier, and (iii) repeat Step (ii) until the search frontier is
empty.

To ensure the algorithm does not revisit isomorphic puzzles, we use
canonical graph representations $[G_p]$ of the puzzle graphs $G_P$.
A canonical graph representation is a binary encoding of a graph with
the property that for any two graphs $G_1, G_2$, $[G_1] = [G_2]$ iff
$G_1 \isom G_2$ (c.f., e.g., \cite{nauty}).  As the search algorithm
runs we record the set $I$ of canonical graph representations $[G_P]$
of each distinct puzzle $P$ that has been added to the search
frontier.  Each time a puzzle $P'$ is considered for being added to
the search frontier we first check whether its canonical graph
representation $[G_{P'}] \in I$, if it is, we do not add $P'$ to the
frontier.  The use of canonical representations of puzzles
dramatically shrinks the search space by searching from $[P]$ rather
than every $P' \isom P$ and by not allowing duplicates of $[P]$ to be
enqueued.  This algorithm SP-BFS is formalized in \autoref{alg:sp-bfs}.

%% A width-$k$ puzzle has potentially $3^k$ distinct rows.  For $j \in
%% [3^k]$ we use $r_j$ to denote the $j^{th}$ row in lexicographic order.
%% For a puzzle $P$ and integer $i \in [|P|]$ we use $P_i$ to denote the
%% $i^{th}$ row of $P$ where the rows are sequenced in lexicographic
%% order.

\begin{algorithm}[t]
  \caption{: Symmetry-Pruned Breadth-First Search}
  \label{alg:sp-bfs}
\begin{algorithmic}[1]
  \Require{An integer $k \ge 0$.}
  \Ensure{The number $b$, which is the size of the largest
    width-$k$ strong USP.}

  \Function{SP-BFS}{$k$}

  \State{Let $Q$ be an empty queue.}
  \State{Let $I$ be an empty set.}
  \State{Let $b = 0$.}
  \State{\Call{enqueue}{$Q, \emptyset$}.}

  \While{$Q$ is not empty}
  \State{$P = \Call{dequeue}{$Q$}$.}
  %% \State{Let $j$ be the index of row $P_{|P|-1}$ in the
  %%   lexicographic order of all rows.}
  \For{$r \in [3]^k \backslash P$}
    %% \For{$r = r_{j+1}$ to $r_{3^k-1}$}
      \State{Let $P' = P \cup \set{r}$.}
      \If{\Call{Verify}{$P'$} and $[G_P'] \not\in I$}
        \State{\Call{enqueue}{$Q, P'$}.}
        \State{$I = I \cup \set{[G_P']}$.}
        \State{$b = |P'|$.}
      \EndIf
    \EndFor
  \EndWhile

  \State{\Return{$b$}.}

  \EndFunction
\end{algorithmic}
\end{algorithm}

We argue the correctness of this algorithm.

\begin{lemma}
  \label{lem:sp-bfs}
For $k \in \mathbb{N}$, \Call{SP-BFS}{$k$} returns the maximum integer
$s$ for which there exists an \susp{s}{k}.
\end{lemma}

\begin{proof}
Ignoring the pruning that $I$ performs for a moment, it is routine to
argue that SP-BFS behaves like a generic breadth-first search
algorithm over the tree of all strong USPs.  This is because of the
downward-closure property of strong USP
(\autoref{lem:downward-closed}), which makes any strong USP $P$
reachable from the trivial strong USP $\emptyset$ using a series of
row inclusions.  \Call{SP-BFS}{$k$} results in an exhaustive search of
all strong USPs of width $k$ and return the maximum size $b$ of such
SUSPs.

We argue that when considering the pruning that $I$ contributes
to, \Call{SP-BFS}{$k$} enqueues exactly one element of each
equivalence class of puzzles that are SUSPs.  Then, as a consequence
of \autoref{lem:SUSP:invariance}, the algorithm must explore every
equivalence class of width-$k$ SUSPs. Hence, it explores an equivalence
class with SUSPs of maximum size and subsequently returns that size,
which is the expected output.

To complete the argument and show that the symmetry pruned search
covers the entire search space of equivalence classes, suppose, for
the sake of contradiction, that there is some smallest $s$ such that
there is an \puz{s}{k} $P$ that does not have its equivalence class
$[P]$ searched.  We know that $s > 1$, because the algorithm starts by
considering all possible \puzs{1}{k}.  Let $P'$ be the \puz{s-1}{k}
created from $P$ by removing one of its rows $r$, $P'$ has as least
one row because $s > 1$.  By hypothesis, the equivalence class of
$[P']$ has been visited by SP-BFS because $P'$'s size is $s - 1 < s$.
Consider $[P]$ and remove the row that corresponded to $r$ to form
$[P]'$.  It must be the case that $[P'] \isom [P]'$.  This isomorphism
extends to $[P]$ in that there must be a row $r'$ such that $([P']
\cup \set{r'}) \isom [P]$, where $r'$ is replaces the row remove from
$[P]$.  Therefore, since $[P']$ is searched, the algorithm must
consider all possible rows to extend by, including $r'$.  This is
means that the equivalence class of $[P]$ is searched, a contradicting
our assumption.  Therefore every equivalence class of SUSPs is
searched by SP-BFS.  \qed
\end{proof}

This approach reduces the size of the search space, improving both the
running time of the search and the space required to keep track of the
frontier puzzles. The worst case running time of SP-BFS is $O(3^k
\cdot \#EQUIV(k) \cdot (T_{\textsc{Verify}}(s_k+1,k) +
T_{\textsc{Canonize}}(s_k, k)),$ where $\#EQUIV(k)$ is the number
equivalence classes of strong USP of width $k$,
$T_{\textsc{Verify}}(s_k+1,k)$ is the time to verify the maximum size
\puzs{s_k+1}{k} examined by the algorithm, and
$T_{\textsc{Canonize}}(s_k,k)$ is the time to compute the canonical
graph representation of each puzzle $P$ considered by the algorithm
(assuming $T_{\textsc{Verify}}$ and $T_{\textsc{Canonize}}$ are
monotone in their parameters).

See \autoref{subsec:symmetry-results} for the experimental results of
running SP-BFS and a discussion of implementation issues.

%% file: upper_bounds.tex
\section{Upper Bounds}
\label{sec:upper-bounds}

Although the main focus of this research line is to construct
sufficiently large strong USP that would imply faster matrix
multiplication algorithms, our techniques and approach can also be
applied to search for tighter upper bounds on the size of strong USP.  We
describe several SUSP-size upper bounds in this section.

\paragraph{$\omega$ Bound.}
Prior work explicitly discusses bounds on the capacity of infinite
families of USP (c.f., \cite[Lemma 3.2, Theorem 3.3]{cksu05}).  Since
every SUSP is a USP, these bounds also apply to SUSP and can be
restated to apply to individual puzzles.  The first bound, which we
denote as the ``$\omega$ bound'', results from (i)
\autoref{lem:omega}, which is monotone non-increasing for fixed $k$,
and (ii) the fact that $\omega \ge 2$. To compute this bound we
evaluate the inequality of \autoref{lem:omega} on increasingly large
$s$ until just before the consequence implies $\omega < 2$ which is in
contradiction with $\omega \ge 2$.

\paragraph{Unique Pieces Bound.}
The second bound, which we denote as the ``unique pieces bound'',
following directly from \autoref{lem:unique-pieces}.  Since that lemma
requires that each row of a (strong) USP have a unique ones, twos, and
threes piece, the total number of rows in a strong USP cannot be more
than $2^k$.

\paragraph{USP Bound.}
The third bound, which we denote as the ``USP bound'', results from
the proof of \cite[Lemma 3.2]{cksu05}.  Although not spelled out in
that article, the proof relies on the following subclaim that
directly bounds $s$ as a function of $k$.
\begin{proposition}
\label{prop:usp-bound}
Let $P$ be a $(s, k)$-USP, then
$$s \le \sum_{c_1 = 0}^k \sum_{c_2 = 0}^{k - c_1} \min\left(\binom{k}{c_1},
\binom{k}{c_2}, \binom{k}{k - (c_1 + c_2)}\right) = O\left(k^2 \cdot \left(\frac{3}{2^{2/3}}\right)^k\right).$$
\end{proposition}
%% However, this is largely because we are using the first half
%% of \autoref{lem:omega} rather than its second half.  If we were using
%% the $\omega$ bound for SUSP that generate infinite families the
%% $\omega$ bound would be much tighter.  This gap parallels the
%% phenomena in \autoref{sec:to-families}.
Note that the USP bound is asymptotically tighter than the unique
pieces bound as $\frac{3}{2^{2/3}} \approx 1.8899 < 2$.

\paragraph{Clique Bound.}
The fourth bound, which we denote as the ``clique bound'', results
from the fact that SUSPs are downward closed
(\autoref{lem:downward-closed}).  In particular if $P$ is an SUSP,
then for every $P' \subseteq P$ with $2$ rows must also be an SUSP.
Fix $k \in \mathbb{N}$ and consider a graph $G_k$ whose vertices
correspond to the possible rows of a width-$k$ puzzle, i.e., strings
in $[3]^k$, and where there is an edge between $r_1, r_2 \in [3]^k$ if
$\set{r_1, r_2}$ is an SUSP.  Observe that by downward closure, each
\susp{s}{k} corresponds to a clique of size $s$ in $G_k$.  This
approach naturally generalizes from the Clique problem to
$h$-HypergraphClique problem where the graph $G_k^h$ consists the same
$3^k$ vertices as $G_k = G_k^2$, but instead has the arity-$h$ edges
$\set{r_1, r_2, \ldots, r_h}$ which are \susps{h}{k}.
\newcommand\SUSP{\mathrm{SUSP}}
\begin{proposition}
\label{prop:clique-bound}
Let $P$ be an \susp{s}{k} and $2 \le h \le s$. Then for $$G^h_k =
\langle V = [3]^k, E = \condset{ P' \sse V }{ P' \text{ is a strong USP and
  } |P'| = h}\rangle,$$ $(G^h_k, s) \in h\text{-HypergraphClique}$.
\end{proposition}
Therefore, the size of a maximum hypergraph clique in $G^h_k$ is an
upper bound of size of width-$k$ SUSP.  We use ``clique bound'' to
denote the specific instantiation of this bound for $h = 2$.
%% (this is true even if we let the $\omega$
%% bound uses puzzle capacity as the family capacity).

\paragraph{Exhaustive Bound.}
For fifth bound, which we denote as the ``exhaustive bound'', we
consider the results of \autoref{alg:sp-bfs} when run in the domain of
$k$ where the full search space can be feasibly explored.  Because
these bounds are based on exhaustive search they are inherently tight.

\paragraph{Downward-Closure Bound.}
The final bound we consider follows from the downward-closure property
of SUSPs.
\begin{proposition}
Let $P$ be an \susp{s}{k} with $k > 1$, then there exists an
\susp{\lceil \frac{s}{3}\rceil}{k - 1}.
\end{proposition}
\begin{proof}
Fix any $c \in [k]$ and consider the $c^{th}$ column of $P$, then, by
averaging, there must be an element of $e \in [3]$ that appears at
least $\lceil \frac{s}{3} \rceil$ times in that column.  Let $P'
\subset P$ be the subpuzzle of $P$ whose rows have $e$ in the $c^{th}$
column.  $P'$ is a strong USP, because $P$ is a strong USP and strong
USPs are downward closed (\autoref{lem:downward-closed}).  Form $P''$
by removing the $c^{th}$ column of $P'$.  $P''$ is a strong USP,
because $P'$ is a strong USP and the strong-USP property is invariant
to addition or removal of constant columns.  By construction, $P''$ is
a \susp{\lceil \frac{s}{3}\rceil}{k - 1}. \qed
\end{proof}
This bound is not as independently applicable like the others, but it
can lift upper bounds of $s \le u$ at $k$ to $s \le 3 u$ at $k+1$.

See \autoref{subsec:results-upper-bounds} for the results of
evaluating the above bounds for small width and a discussion of issues
involved in concretely calculating them.

%% file: implementation.tex
\section{Implementation}
\label{sec:implementation}

We implemented our verification algorithms, heuristics, and search
algorithms, along with various utilities and appropriate
datastructures to represent underlying information such as puzzles in
C++. The source code for our implementation is available under a MIT
License at \url{https://bitbucket.org/paraphase/matmult}.

We use a number of external libraries with subroutines that are key to
the functioning of our algorithms.  Our IP-based verifier and Clique
bound calculator both use the commercial, closed-source mixed-integer
programming solver Gurobi to solve the integer programs produced by
our reductions \cite{gurobi}.  Our SAT-based verifier uses, by
default, the \texttt{kissat-sc2021-sat} solver from the 2021 SAT
Competition by A.~Biere, M.~Fleury, and M.~Heisinger \cite[page
  10]{SAT2021}.  Note that the conference version of this article used
the MapleCOMSPS solver---see \autoref{sec:choice-sat} for a discussion
of solver benchmarks, comparisons, and choice.  We implemented
\autoref{alg:sp-bfs} using our hybrid verifier, and the graph
automorphism library Nauty \cite{nauty} as a subroutine to perform the
required graph canonization on $G_P$.  The original versions of our
SP-BFS implementation targeted a high-performance computing cluster
environment, because our brute force and dynamic programming
implementations were not efficient enough.  Subsequent improvements to
our verification algorithms made this unnecessary.  Despite this, our
SP-BFS implementation is still in MPI and uses a MapReduce framework
\cite{mrmpi} to maintain a distributed search frontier.

Our code base also contains multiple implementations of
depth-first-search-inspired algorithms for locating strong USPs.
These algorithms use our hybrid verification implementation and puzzle
symmetry pruning technique discussed in \autoref{sec:search}.  For
brevity and to keep this article focused on strong-USP verification,
we elect not to discuss these algorithms and defer them to a
subsequent article.  That said, some of the concrete puzzles we found
and report in the next section were generated by such algorithms.
These puzzles once found were experimentally verified as strong USPs
using the techniques discussed in detail in \autoref{sec:verify}.

%% file: results.tex
\section{Experimental Results}
\label{sec:results}

Our experimental results come in several flavors for small-constant
width $k$: (i) constructive lower bounds on the maximum size of
width-$k$ strong USPs witnessed by found puzzles, (ii) upper bounds on
the maximum size of width-$k$ strong USPs, (iii) the number of SUSPs
and SUSP equivalence classes for width $k$, (iv) experimental data
comparing the run times of our verification algorithms and
distinguishing likelihood of our heuristics, and (v) a benchmark data
set of SAT/UNSAT instances that we use to compare the effectiveness of
competitive SAT solvers as subroutines for the SAT-based part of our
verifier.

All of the results in this section were produced by running our
algorithm implementations on the same Ubuntu 20.04 PC with a 3.00 GHz
Intel Core i9-10980XE CPU and 128 GB of RAM.

\subsection{New Upper and Lower Bounds on the Size of Strong USPs}
\label{subsec:usps_found}
\label{subsec:results-upper-bounds}
\label{subsec:symmetry-results}

\paragraph{New Lower Bounds.}
\autoref{table:compare} summarizes new lower bounds for maximum SUSP
size in comparison with \cite{cksu05}.  The lower bounds of
\cite{cksu05} are from the constructions in their Propositions 3.1 and
3.8, which give families of strong USPs for even $k$ or $k$ divisible
by three.  For $k$'s which are not divisible by two or three, we
extrapolate their construction by adding a new column, this preserves
the SUSP property.  The upper bounds on $\omega$ in this table are
computed by plugging $s$ and $k$ into \autoref{lem:omega} and
optimizing over $m$.  For clarity we omit $\omega$'s that would be
larger than previous columns.  Our results in this table we produced
by running SP-BFS and other search algorithms which verify that the
final result is a strong USP.  Our bounds are tight for all $k \le 5$,
because of the exhaustive nature of SP-BFS, and constructively
improve the known lower bounds for $4 \le k \le 12$.

\begin{table}[t]

  \begin{centering}
    \hspace{-1.5ex}
    \begin{tabular}{ccrrrrrrrrrrrr}
      \toprule
      &&\multicolumn{12}{c}{$k$} \\ \cmidrule{3-14}
      
      & & \multicolumn{1}{c}{~1} & \multicolumn{1}{c}{~2} &
      \multicolumn{1}{c}{~3} & \multicolumn{1}{c}{~4} &
      \multicolumn{1}{c}{~~5~~} & \multicolumn{1}{c}{~6} &
      \multicolumn{1}{c}{~7} & \multicolumn{1}{c}{~8} &
      \multicolumn{1}{c}{~9} &\multicolumn{1}{c}{~10} &
      \multicolumn{1}{c}{~11} & \multicolumn{1}{c}{~12}\\ \hline
      
      \multirow{2}{*}{\cite{cksu05}} &$s \ge$ & \textbf{1} &
      \textbf{2} & \textbf{3} & 4 & 4 & 10 & 10 & 16 & 36 & 36 & 36 &
      136 \\
      
      &$\omega \le$ & 3.00 & 2.88 & 2.85 & 2.85 & & 2.80 & & & 2.74 &
      & & 2.70 \\
      
      \multirow{2}{*}{This work}&$s \ge$ & \textbf{1} & \textbf{2} &
      \textbf{3} & \textbf{5} & \textbf{8} & 14 & 21 & 30 & 42 & 64 &
      112 & 196\\
      
      &$\omega \le$ & ~3.00 & ~2.88 & ~2.85 & ~2.81 & ~2.78 & ~2.74 &
      ~2.73 & ~2.72 & ~2.72 & ~2.71 & ~2.68 & ~2.66 \\ \bottomrule
      
    \end{tabular}
  \end{centering}
  \medskip
  \caption{Comparison with \cite{cksu05} of lower bounds on the
    maximum of size of width-$k$ strong USPs and the upper bounds on
    $\omega$ they imply. Bold font indicates tight results for that
    $k$. \label{table:compare}}
\end{table}

\autoref{fig:examples} contains representative examples of
maximal-size strong USPs we found for $k \le 6$.
\begin{figure}[t]
  \begin{center}
  \includegraphics[width=.8\linewidth]{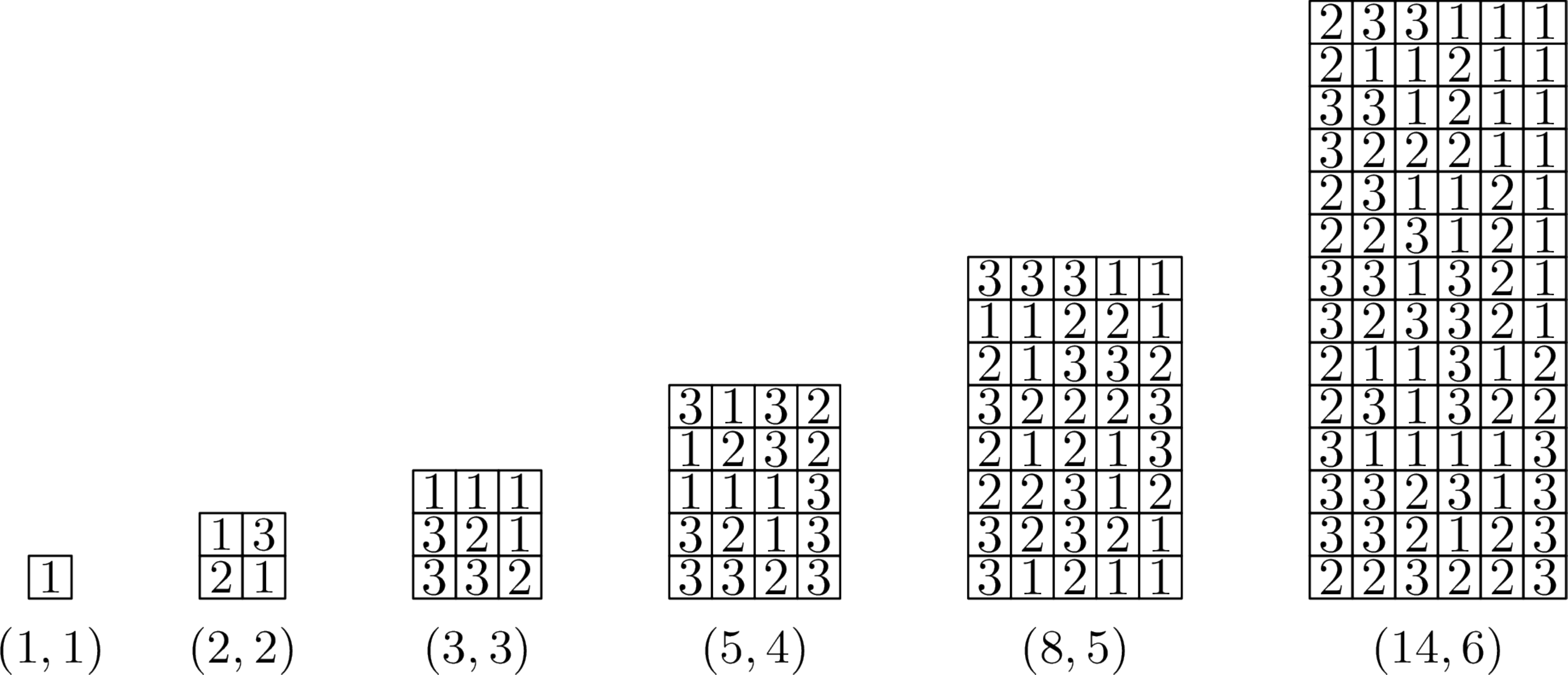}
  \caption{Representative maximal-size strong USPs found for width $k
    = 1, 2, \ldots, 6$.}
  \label{fig:examples}
  \end{center}
\end{figure}
The strong uniquely solvable \puzs{14}{6} we found represent the
greatest improvement in $\omega$ versus the construction of
\cite{cksu05} for small $k$.  Further, our puzzle for $k = 12$ is the
result of taking the Cartesian product of two copies of a strong
uniquely solvable \puzs{14}{6}.  Note that Proposition 3.8 of
\cite{cksu05} gives an infinite family of strong USPs that achieves
$\omega < 2.48$ as $k$ goes to infinity, which is stronger than our
results are directly able to achieve. 

\paragraph{New Upper Bounds.}
\autoref{table:upper-bounds} summarizes the results of evaluating the
bounds from \autoref{sec:upper-bounds} for puzzles of width $k \le
12$.  The calculations were routine except for the clique bound that
required constructing $G_k$, converting it into a mixed integer
program, and solving that program using Gurobi \cite{gurobi}.  This
was feasible on our test system up to $k = 11$. We also experimented
with calculating the upper bounds for the $3$-HypergraphClique bound,
but found it infeasible to compute for $k \ge 5$ and so have omitted
the results.  The final row of the table contains the best upper
bounds we achieved, including applying the downward-closure bound to
lift adjacent bounds at $k = 6$ and $k = 12$.  These upper bounds are
stronger than those immediately implied by \cite{cksu05}.

Observe that exhaustive search produced the best and tightest bounds,
and that the clique bound is considerably stronger than the unique
pieces, USP, and $\omega$ bounds.  The unique pieces bounds appears to be
stronger than the USP bound, but we know that that is an artifact of
the small value of $k$.  As $k$ increase, the USP bound will become
tighter than the unique pieces bound.
\begin{table}[t]
  \begin{center}
    \begin{tabular}{crrrrrrrrrrrr}
      \toprule
      & \multicolumn{12}{c}{$k$} \\ \cmidrule{2-13}
      
      Bound & \multicolumn{1}{c}{~1} & \multicolumn{1}{c}{~2} &
      \multicolumn{1}{c}{~3} & \multicolumn{1}{c}{~4} &
      \multicolumn{1}{c}{~~5~~} & \multicolumn{1}{c}{~6} &
      \multicolumn{1}{c}{~7} & \multicolumn{1}{c}{~8} &
      \multicolumn{1}{c}{~9} &\multicolumn{1}{c}{~10} &
      \multicolumn{1}{c}{~11} & \multicolumn{1}{c}{~12}\\ \hline
      
      $\omega$ & ~3 & ~7 & ~15 & ~31 & ~62 & ~120 & ~230 & ~438 & ~831
      & ~1,575 & ~2,890 & ~5,637\\

      Unique & 2 & 4 & 8 & 16 & 32 & 64 & 128 & 256 & 512 & 1,024 & 2,048 & 4,096 \\
      
      USP & 3 & 6 & 12 & 24 & 45 & 87 & 168 & 312 & 597 & 1,140 &
      2,112 & 4,023\\
      
      Clique & \textbf{1}& 3 & 5 & 9 & 17 & 30 & 55 & 105 & 186 & 348
      & 654 &\\
      
      Exhaustive& \textbf{1}& \textbf{2}& \textbf{3}& \textbf{5}&
      \textbf{8} &&&&&&&\\ \hline
      
      Best & \textbf{1}& \textbf{2}& \textbf{3}& \textbf{5}&
      \textbf{8} & 24 & 55 & 105 & 186 & 348 & 654 &
      1,962\\ \bottomrule
    \end{tabular}
  \end{center}
  \caption{Upper bounds on the size of SUSPs for widths $k \le 12$.
    Bold font indicates the bound is tight, and blanks indicate the
    calculation for this puzzle width was
    infeasible.  \label{table:upper-bounds}}
\end{table}
Based on the processing time we spent on $k=6$, we conjecture that $s
= 14$ is tight for $k = 6$ and that our lower bounds for $k > 6$ are
not.  Our results suggests there is considerable room for improvement
in the construction of strong USPs, and that it is possible that there
exist large puzzles for $k = 7, 8, 9$ that would beat \cite{cksu05}'s
constructions and perhaps come close to the Coppersmith-Winograd
refinements.  That said, it seems that new insights into the SUSP
search problem are required to proceed for $k > 6$.

\paragraph{Counting Strong USP.}
\autoref{table:susp-equiv} shows the number of strong USPs and
equivalence classes of SUSP exhaustively calculated using SP-BFS with
and without symmetric pruning.  Observe that the number of strong USPs
is many orders of magnitude more than the number of equivalence
classes of strong USPs, even for \susps{3}{3}.  Exhaustive search
became infeasible even with puzzle symmetry pruning for $k \ge 6$ as
the memory usage of \autoref{alg:sp-bfs} for storing the search
frontier exceeds the 128GB available on our test system.

\begin{table}[t]
  \begin{centering}
    \hspace{-6.5ex}
    \begin{tabular}{lcrrcrrcrrcrrcrrcrr}
      \toprule
      & \multicolumn{18}{c}{$k$} \\ \cmidrule{3-19}
      
      $s$ & & \multicolumn{2}{c}{1} &~~& \multicolumn{2}{c}{2} &~~&
      \multicolumn{2}{c}{3} &~~& \multicolumn{2}{c}{4} &~~&
      \multicolumn{2}{c}{5} &~~& \multicolumn{2}{c}{6}
      \\ \cmidrule{1-1} \cmidrule{3-4} \cmidrule{6-7} \cmidrule{9-10}
      \cmidrule{12-13} \cmidrule{15-16}\cmidrule{18-19}
      
      %% $s$ && \multicolumn{1}{c}{E} & \multicolumn{1}{c}{\#} &&
      %% \multicolumn{1}{c}{E} & \multicolumn{1}{c}{\#} &&
      %% \multicolumn{1}{c}{E} & \multicolumn{1}{c}{\#} &&
      %% \multicolumn{1}{c}{E} & \multicolumn{1}{c}{\#} &&
      %% \multicolumn{1}{c}{E} & \multicolumn{1}{c}{\#} &&
      %% \multicolumn{1}{c}{E} & \multicolumn{1}{c}{\#} \\ \hline
      
      1 && \textbf{1} & 3 && \textbf{2} & 9 && \textbf{3} & 27 &&
      \textbf{4} & 81 && \textbf{5} & 243 && \textbf{7} & 729 \\
      
      2 && &&& \textbf{2} & ~24 && \textbf{9} & 408 && \textbf{33} &
      4,848 && \textbf{91} & 50,160 && \textbf{229} & 486,024 \\
      
      3 && &&& &&& \textbf{9} & ~1,800 && \textbf{240} & 182,304 &&
      \textbf{2,429} & 8,361,000 && \textbf{16,971} & ~291,347,280 \\
      
      4 && &&& &&& &&& \textbf{728} & ~2,445,120 && \textbf{59,149} &
      ~992,377,400 && \textbf{1,611,648} & ? \\
      
      5 && &&& &&& &&& \textbf{190} & 3,248,640 && \textbf{707,029} &
      ? && ? & ? \\
      
      6 && &&& &&& &&& &&& \textbf{2,337,715} & ? && ? & ? \\

      7 && &&& &&& &&& &&& \textbf{1,359,649} & ? && ? & ? \\

      8 && &&& &&& &&& &&& \textbf{89,196} & ? && ? & ? \\

      9 && &&& &&& &&& &&& &&& ? & ? \\ \bottomrule
      
    \end{tabular}
  \end{centering}
  \medskip
  \caption{Number of equivalence classes (bold face, left) versus
    total number of encoded SUSPs (normal face, right) by \puz{s}{k}
    dimensions. Computed using \autoref{alg:sp-bfs}. Empty cells
    indicate that the number of SUSPs and equivalence classes is zero.  ?'s
    indicate unknown values that were infeasible to compute. \label{table:susp-equiv}}
\end{table}

\subsection{Algorithm Performance}
\label{subsec:performance}

To measure the performance of our verification algorithms and
heuristics we ran them on 10,000 random puzzles at each point on a
sweep through parameter space for widths $k = 5 \ldots 12$ and sizes
$s = 1 \ldots 60$.  We chose to test performance via random sampling
because we do not have access to a large set of solved instances.
This domain coincides with the frontier of our search space, and we
tuned the parameters of the heuristics and algorithms in the hybrid
algorithm to perform well in this domain.  We did not deeply
investigate performance characteristics outside of this domain.  In
Figures~\ref{fig:verification-means}, \ref{fig:heuristic-success}, \&
\ref{fig:hybrid-time} we plot results, for brevity, that are
representative of the parameter space only for $k \in \set{6, 9}$.

\begin{figure}[t]
  \includegraphics[width=\linewidth]{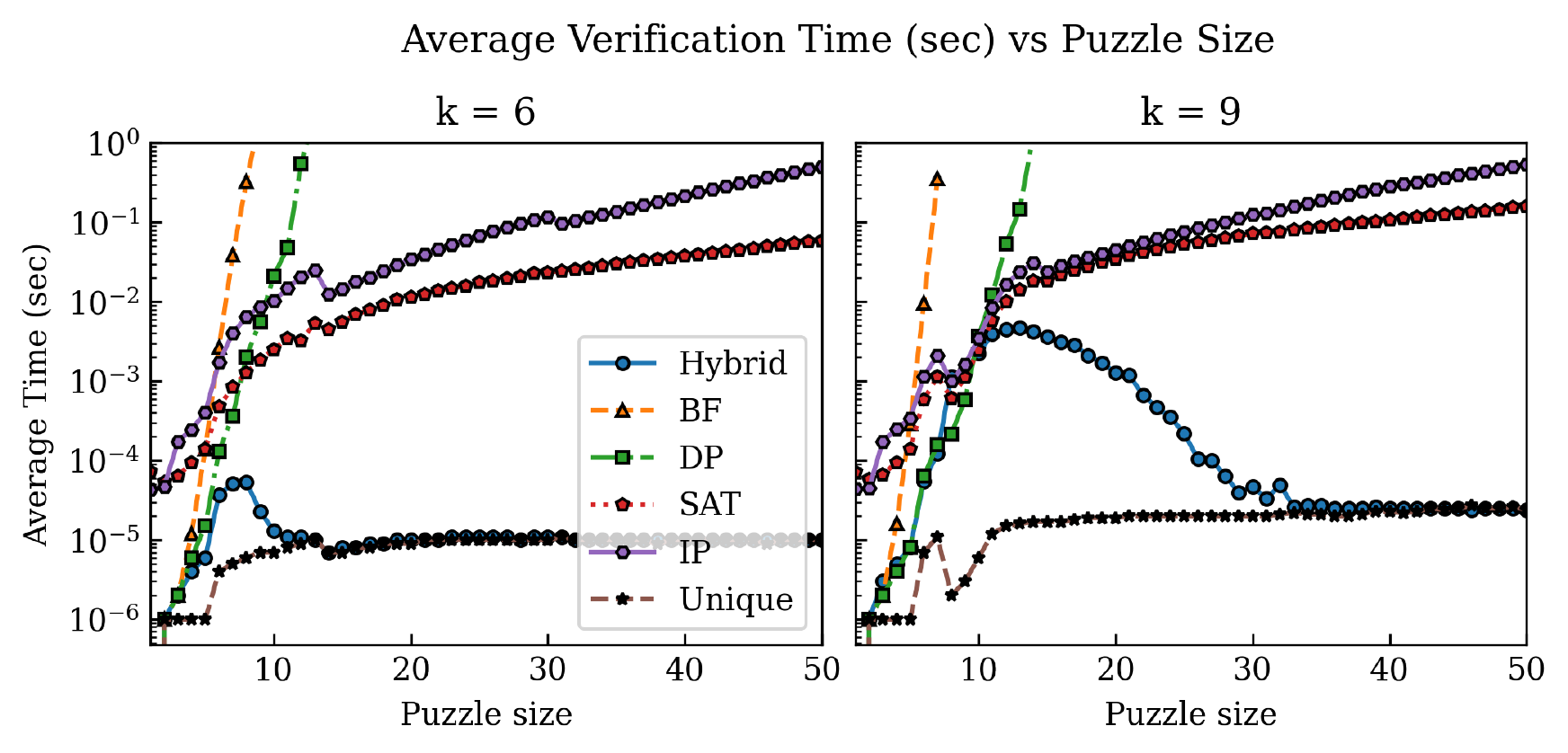}
  \caption{Log plots of the average running times for verifying 10,000
    random \puzs{s}{k} for each $s \in [50], k \in \set{6, 9}$. The
    plots describe the behavior of five verification algorithms brute
    force (BF), dynamic programming (DP), reduction to satisfiability
    (SAT), reduction to integer programming (IP), and our hybrid
    algorithm (Hybrid).  The running time of the unique pieces heuristic
    is also included. \label{fig:verification-means}}
\end{figure}

\paragraph{Running Time.}
\autoref{fig:verification-means} shows the average running times of
our verification algorithms in seconds.  The brute force and dynamic
programming algorithms perform poorly except for very small size, $s
\le 8$, and their curves loosely match the exponential-time bounds
we expect.  The plots for the two reduction-based algorithms (SAT and
IP) behave similarly to each other.  They are slower than brute force
and dynamic programming for small values of $s$, and their behavior
for large $s$ is quite a bit faster.  We speculate that the former is
due to the cost of constructing the reduced instance and overhead of
the third party tools.  Further observe that the SAT reduction handily
beats the IP reduction on large size for $k = 6$, but as $k$
increases, the gap decreases.  We also note that across the settings
of $k$ the IP reduction has effectively the same running time and is
independent of $k$.  This is likely because the size of the IP
instance depends only on $s$.  The hybrid algorithm generally performs
best or close to best at small values of $s$ and is clearly faster for
large values of $s$.  Notice that it matches the dynamic programming
algorithm closely for small values of $s$ and then diverges when the
reduction-based algorithms and heuristics are activated at
larger $s$.  Observe that the hybrid algorithm is effectively constant
time for large $s$, though the size for which this happens increases
as a function of $k$.  We expect this is because the density of strong
USPs decreases rapidly with $s$, and that the randomly selected
puzzles are likely far from satisfying \autoref{def:strong-USP} and,
hence, they are quickly rejected by the unique pieces heuristics.
Further evidence of this is that running time of the hybrid algorithm
converges to the running time of the unique pieces heuristic for large
$k$.

\begin{figure}[t]
  \includegraphics[width=\linewidth]{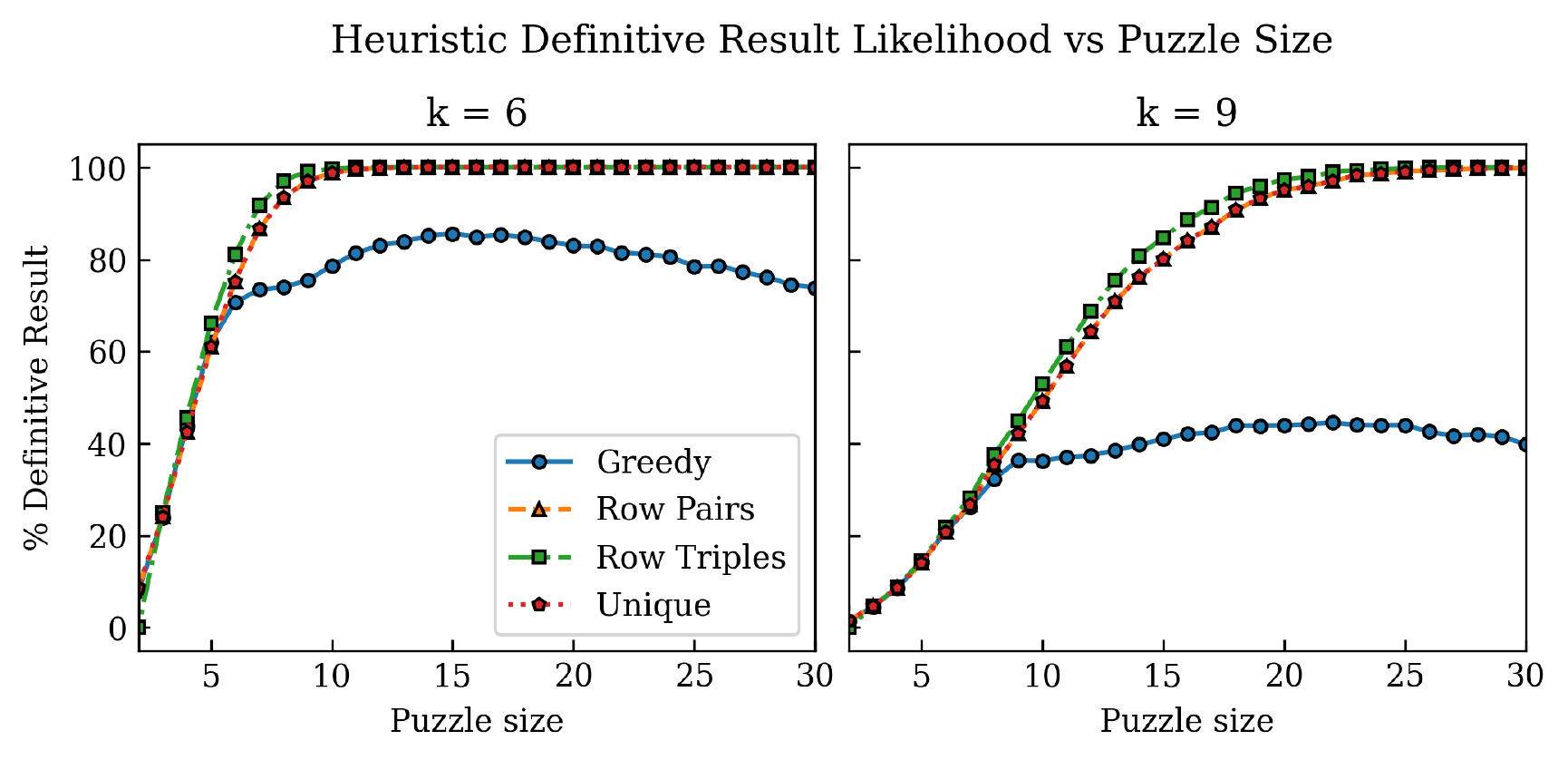}
  \caption{Plots of the likelihood that each of the heuristics produces a
    definitive results on 10,000 random \puzs{s}{k} for each size $s \in
    [50]$ and width $k \in \set{6, 9}$. Here ``row pairs'' is
    $\textsc{HeuristicDownwardClosed}(P, 2)$ and ``row triples'' is
    $\textsc{HeuristicDownwardClosed}(P, 3)$.  The row pairs points
    are plotted, but are hard to see, because the unique pieces points
    coincides with them.
    \label{fig:heuristic-success}}
\end{figure}

\paragraph{Heuristic Effectiveness.}
\autoref{fig:heuristic-success} shows the probability that each
individual heuristic distinguishes a random puzzle in our benchmark.
Observe that the distinguishing power of the downward closure
heuristic for $s' = 2$ and unique pieces heuristics coincide,
demonstrating experiment consistency with
\autoref{lem:heuristic-equiv}. Further, and for the same reason, the
downward closure heuristic for $s' = 3$ has at least as high a
distinguishing likelihood as the unique pieces heuristic.  In the
plots, these three heuristics achieve almost 100\% probability of
distinguishing random puzzles by size $s = 30$. The greedy heuristic
perform less well than the others and get substantially worse as $k$
increases.  We do not plot the running times of the heuristics here,
but they behave as expected by the earlier analysis.  As we noted
earlier, unique pieces is linear time in the size of the puzzle and
the fastest of the heuristics.  \autoref{fig:verification-means} shows
how the running time of the hybrid algorithm and unique pieces
converges as essentially all random puzzles of large size, which the
benchmark examined, are verified as non-SUSPs by this heuristic.

\begin{figure}[t]
  \includegraphics[width=\linewidth]{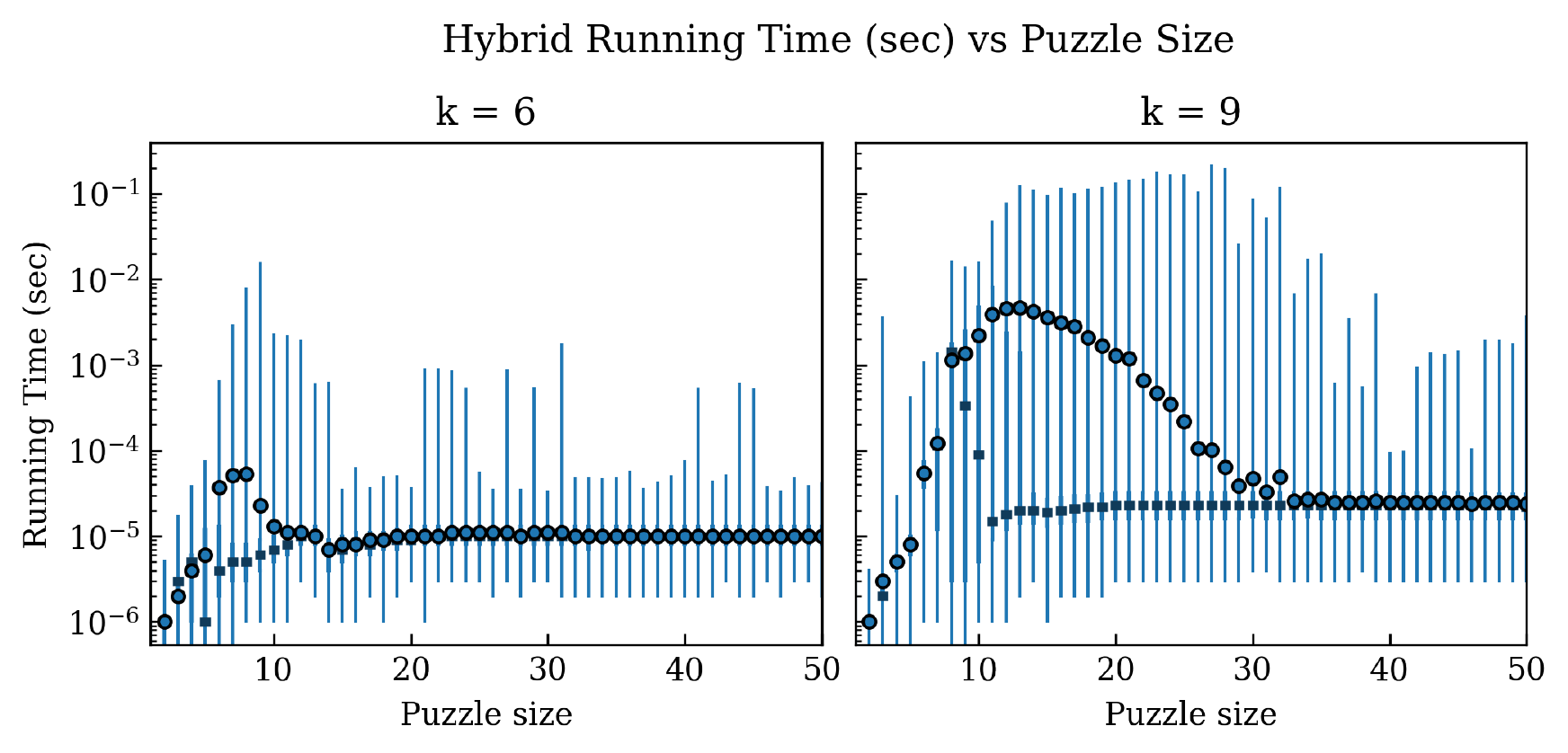}
  \caption{Log box plots of the distribution of the running times of
    the hybrid verification algorithm on 10,000 random \puzs{s}{k}
    for each $s \in [50], k \in \set{6, 9}$.  The blue circles
    denote the average running times of the hybrid algorithm.  The dark
    blue blocks indicates the median times. The thick vertical lines
    indicate the middle 50\% of times, and the thin vertical lines
    indicate the full range of running times at each $s$.
    \label{fig:hybrid-time}}
\end{figure}

\paragraph{Variation in Running Time.}
Finally, we look at the variation in the running times of the hybrid
algorithm in \autoref{fig:hybrid-time}.  For small $s$, the running
time distribution is far from a normal distribution--the average is
far above the median and middle 50\% of running times.  This effect
becomes even more pronounced as $k$ increases. However, we find that
as $s$ increases, the median running time converges with the median
running time of the unique pieces heuristic, and then for larger $s$,
the average running time converges as well.  This is a consequence of
the hybrid algorithm having to run the orders of magnitude slower
reduction-based algorithms when the fast heuristics fail to resolve the
instance.  Although not plotted here, we found that the range of the
distribution of running times for the SAT-based verifier was larger
than for the IP-based verifier, even though the IP-based verifier was
slower on average.

Overall, our hybrid verification algorithm performs reasonably well in
practice on random instances, despite reductions through
\NP{}-complete problems.

\subsection{Choice of SAT Solver}
\label{sec:choice-sat}

In the conference version of this article we examined only one SAT
solver for use in our implementation, MapleCOMSPS, a conflict-driven
solver that uses a learning rate branching heuristic, and that was a top
performer at the 2016 SAT Competition \cite{bhj17,lgpc16,ajx20}.  In
this article we create a set of benchmark satisfiability instances,
using the SUSP verification reduction on a variety of puzzles (recall
\autoref{subsec:sat}), and examined the performance of
35\footnote{There were 39 SAT solvers submitted to the main track.  We
  use the default build configuration for each submission.  We were
  unable to build three of them, and one that builds repeatedly
  crashed on all benchmarks without producing a result. We tested the
  remaining 35.} solvers submitted to the main track of the 2021 SAT
Competition \cite{SAT2021}.

We select benchmark instances consisting of \puz{s}{k} with sizes from
the set $$\set{(2, 2), (3, 3), (5, 4), (8, 5), (14, 6), (21, 7), (30,
  8), (42, 9)}.$$ We choose these sizes, because we want positive and
negative instances and these sizes represent the largest strong USPs
of each width we have been able to locate through search.  For each
size we created ten puzzles that are strong USPs and ten puzzles that
are not.  To create the ten non-SUSPs we randomly generated a puzzle
of that size and verified it was not a strong USP.  To create the ten
strong USPs we for each size we used the results of our search
algorithms.  Then we ran all of the puzzles through our SAT reduction
to create \texttt{.dimacs} files for each instance.  Note that the
SUSPs correspond to UNSAT instances and non-SUSPs correspond to SAT
instances.  In total there are 160 instances in this benchmark.  We
then ran each of the 35 solvers on each the 160 instance files and
check the output of each run against the expected result. For each
trial, we record the user CPU time reported by the Linux \texttt{time}
command, or a timeout if the program runs more than 5000 seconds
without halting (mimicking the rules of the real SAT competition).
For comparison, we also run the MapleCOMSPS solver (from earlier
version of this article), our MIP-based verifier (recall
\autoref{subsec:mip}) and our final hybrid verification algorithm on
the same set of benchmark puzzles.

To compare the results of each solver we calculate the maximum time to
complete each instance across all of the runs, which is 5000 seconds if a
run timed out, and then divide by that maximum time to normalize all of the
running times to the interval $[0,1]$.  We calculate a benchmark score
for each solver by summing their relative running times across all
instances.  \autoref{table:solvers-perf} contains the benchmark scores
for each solver.

\begin{table}[t]
  \begin{center}
    %Solver Benchmarking Scores
    \begin{tabular}{lrrrr}
      \toprule
      Solver & ~~~~SAT & ~~UNSAT & ~~~~~Total & ~~Timeouts\\ \hline
      cadical-hack-gb & 17.51 & 15.97 & 33.48 & 15 \\ 
      cadical-less-UP & 19.81 & 16.14 & 35.95 & 15 \\ 
      cadical-PriPro & 19.49 & 15.62 & 35.11 & 15 \\ 
      cadical-PriPro\_no\_bin & 16.55 & 15.73 & 32.28 & 15 \\ 
      cadical-rp & 19.08 & 15.78 & 34.85 & 15 \\ 
      cadical-sc2021 & 18.82 & 16.80 & 35.62 & 16 \\ 
      Cadical\_SCAVEL01 & 33.49 & 16.73 & 50.23 & 15 \\ 
      Cadical\_SCAVEL02 & 40.97 & 27.28 & 68.26 & 15 \\ 
      cleanmaple & 30.44 & 18.93 & 49.37 & 17 \\
      CleanMaple\_PriPro & 30.70 & 20.18 & 50.87 & 18 \\ 
      hCaD & 19.70 & 16.52 & 36.22 & 16 \\ 
      hKis & 13.15 & 17.30 & 30.45 & 16 \\ 
      kissat\_bonus & 13.04 & 16.59 & 29.63 & 15 \\ 
      kissat\_cf & 12.06 & 16.19 & 28.26 & 14 \\ 
      kissat\_gb & 12.52 & 17.27 & 29.79 & 17 \\ 
      kissat-MAB & 15.28 & 16.07 & 31.36 & 15 \\ 
      kissat-sat\_crvr\_gb & 13.37 & 16.64 & 30.01 & 16 \\ 
      kissat-sc2021 & 12.32 & 16.08 & 28.40 & 14 \\ 
      kissat-sc2021-sat & \textbf{12.02} & 16.06 & \textbf{28.08} & 14 \\ 
      kissat-sc2021-sweep & 12.82 & 16.24 & 29.07 & 16 \\ 
      lstech\_maple & 15.13 & 14.83 & 29.96 & 12 \\ 
      Maple\_MBDR\_BJL6\_Tier2 & 19.46 & 16.02 & 35.47 & 14 \\ 
      Maple\_MBDR\_BJL7\_Local & 19.98 & 15.49 & 35.47 & 13 \\ 
      Maple\_MBDR\_Cent\_PERM\_10K & 25.20 & 15.96 & 41.16 & \textbf{12} \\ 
      Maple\_MBDR\_Cent\_PERM\_75K & 25.07 & 16.00 & 41.06 & \textbf{12} \\ 
      Maple\_simp21 & 12.53 & 16.72 & 29.26 & 15 \\ 
      MapleSSV & 15.56 & 16.68 & 32.24 & 16 \\
      parafrost-nomdm-sc2021 & 18.11 & 15.56 & 33.67 & 14 \\ 
      parafrost-sc2021 & 24.15 & 15.61 & 39.76 & 14 \\
      Relaxed\_LCFTP & 12.80 & 17.55 & 30.35 & 16 \\
      Relaxed\_LCFTP\_V2 & 13.97 & 16.17 & 30.14 & \textbf{12} \\ 
      Relaxed\_LCMDCBDL\_BLB & 15.38 & 15.95 & 31.33 & 14 \\ 
      Relaxed\_LCMDCBDL\_SCAVEL01 & 13.95 & 16.08 & 30.03 & 15 \\ 
      Relaxed\_LCMDCBDL\_SCAVEL02 & 25.45 & 79.43 & 104.88 & 17 \\
      slime & 17.26 & \textbf{14.73} & 31.99 & 13 \\ \hline 
      MapleCOMSPS & 12.98 & 17.42 & 30.40 & 16 \\ 
      Gurobi & 30.20 & 0.00 & 30.20 & 0 \\ 
      Hybrid & 0.00 & 0.01 & 0.01 & 0 \\ \bottomrule
    \end{tabular}
  \end{center}
  \caption{Scores for solvers on our SUSP verification benchmark.  The
    SAT and UNSAT score are out of 80, the total score and timeouts
    are out of 160.  Lower scores are better and minimum values for
    each SAT solver are bold in each column. The top part of the
    table includes the SAT solvers we tested from the 2021 SAT
    Competition \cite{SAT2021}. \label{table:solvers-perf}}  
\end{table}

MapleCOMSPS, the solver we used in the conference version of this
article, performs similarly to the best scoring solvers from the 2021
competition.  The recorded timeouts across all solvers come almost
exclusively from the UNSAT instances derived from \susps{30}{8} and
\susps{42}{9}.  The Gurobi-based verifier performs substantially worse
than the best performing satisfiability solvers on SAT instances
(non-SUSPs), but dramatically better on UNSAT instances (SUSPs). 

\begin{figure}[t]
  \includegraphics[width=\linewidth]{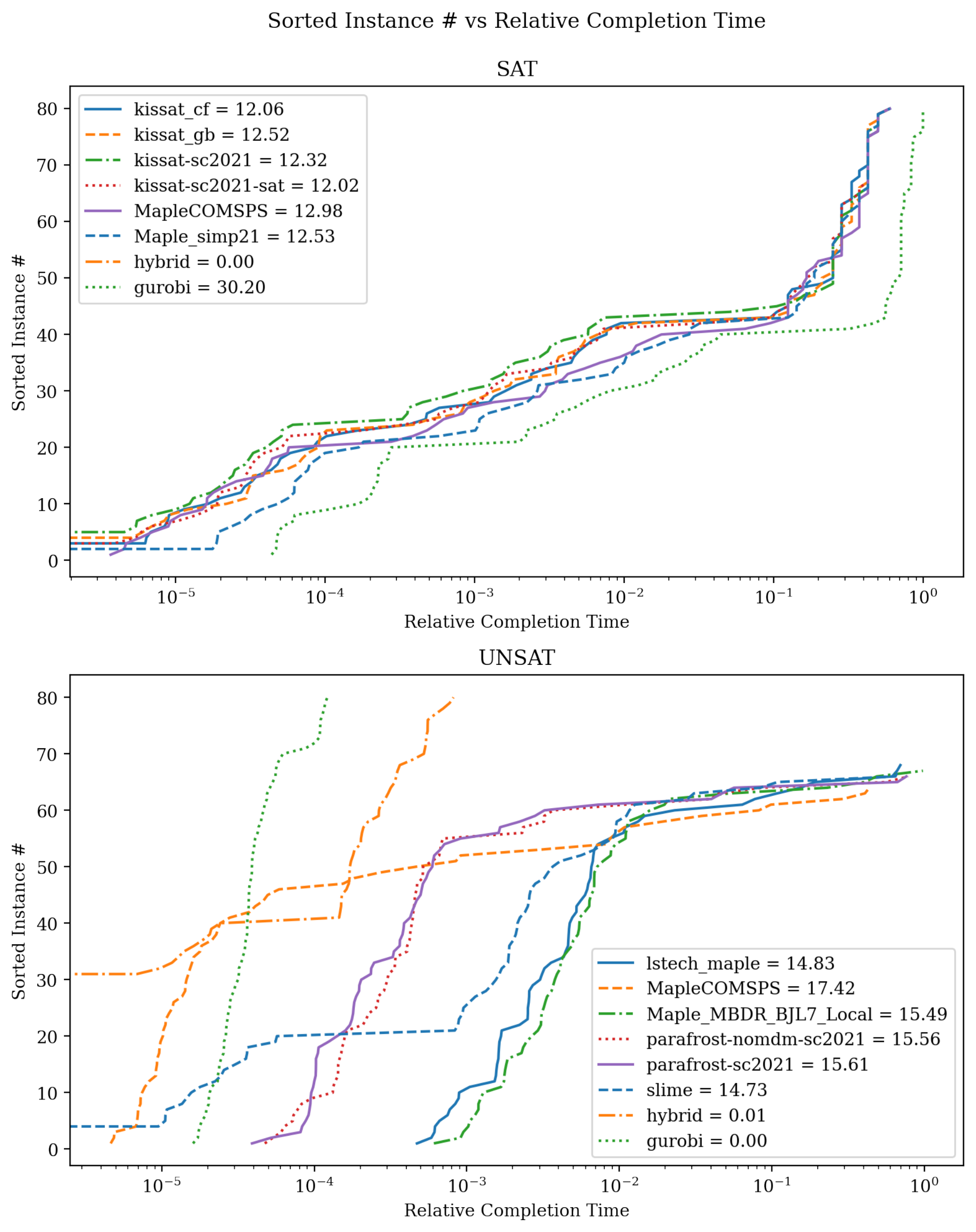}
  \caption{Plots of the sorted relative completion times for SAT and
    UNSAT instances on the five best-scoring solvers for that instance
    type. \label{fig:solvers-perf}}
\end{figure}

\autoref{fig:solvers-perf} shows the performance of the
Gurobi-based verifier against the five solvers with the best SAT
scores.  In this plot the instance completion times for each solver
are sorted in increasing order, so that curves further to the left are
better.  If this were not a log-plot, the area to the left of the
curve would be proportional to the benchmark scores from
\autoref{table:solvers-perf}.  Observe that for SAT instances, the SAT
solvers, including MapleCOMSPS, follow similar trajectories.  Gurobi
performs an order of magnitude worse across all SAT instances.  The
hybrid algorithm, although plotted, is not visible because of how
effective the heuristics are at identifying random SAT (non-SUSP)
instances.  For UNSAT instances, the situation is different.  Gurobi
performs relatively more slowly for small, easier instances, but
substantially better than the SAT solvers for larger, harder
instances.  The performance of the solvers on easier UNSAT instances
is more varied than the corresponding case for SAT instances, but this
does not translate into much of a difference in benchmark score because
the magnitude of the relative completion time is low.

For UNSAT instances, the benchmark score is dominated by the number of
timeouts, each of which effectively adds one to the score.  Indeed,
the plots for the SAT solver cut off between instance numbers 60 to
70, because the remaining instances cause timeouts.  Finally, notice
that hybrid algorithm out performs the others for small UNSAT
instances -- these are instances of the sort where the brute force and
bi-directional search algorithms are applied.  For larger instances
the hybrid algorithm tracks an order of magnitude worse than the
Gurobi-based verifier.  This is because our algorithm is tuned to
encounter many more SAT instances (non-SUSPs) than UNSAT instances
(SUSPs).  Further, because the one-sided heuristics rule out SAT
instances quickly in practice, on UNSAT instances the hybrid algorithm
runs these heuristics first, but then has to fall back on the
Gurobi-based verifier causing some overhead.

Ultimately, the results of these benchmarking experiments suggest that
there is not a substantial difference between using the 2016
MapleCOMSPS and the best solvers from the 2021 competition. Even so,
we choose \texttt{kissat-sc20221-sat} as the default solver in our
implementation, because it performed the best on our benchmark of SAT
instances.  Using our current approach, Gurobi is essential to the
feasible verification of SUSPs.

The benchmark instances and puzzles, and the entirety of the raw
timing data can be found in our
repository\footnote{\url{https://bitbucket.org/paraphase/matmult/src/main/data_set/}}.

%% file: conclusion.tex
\section{Conclusions}
\label{sec:conclusion}

We initiated the first study of the verification of strong USPs and
developed practical software for both verifying and searching for
them.  We give tight results on the maximum size of width-$k$ strong
USPs for $k \le 5$ and improved upper and lower bounds on maximum
strong-USP size for $k \le 12$.  We prove a number of properties of
strong USPs related the verification and search.  We also produce a
new set of benchmark instances for SAT solvers.

Although our results do not produce a new upper bound on the running
time of matrix multiplication, they demonstrate there is promise in
this approach.  There are a number of open questions.  Is strong-USP
verification \coNP{}-complete?  What is the maximum strong-USP
capacity?  Is there a way to bridge the apparent gap between the
values of $\omega$ implied by single SUSPs and the values implied by
infinite families of SUSPs?  What are tight bounds on maximum-size
strong USPs for $k \ge 6$ and do these bound lead to asymptotically
faster algorithms for matrix multiplication?

The main bottleneck in our work is the size of the search space---new
insights seem to be required to substantially reduce it.  Are there
subclasses of strong USPs that can be more effectively searched?  Are
there search strategies that would be more effective on this space?

%% file: acknowledgments.tex
\section*{Acknowledgments}
The authors thank the anonymous reviewers for their detailed and
thoughtful suggestions for improving this work.

The second and third authors thank Union College for the Undergraduate
Summer Research Fellowships funding their work.  The first author
thanks the many undergraduate students that have contributed in some
form to this project over the years, including: Jonathan Kimber,
Akriti Dhasmana, Jingyu Yao, Kyle Doney, Quoc An, Harper
Lyon, Zachary Dubinsky, Talha Mushtaq, Jing Chin, Diep Vu, Hung
Duong, Vu Le, Siddhant Deka, Baibhav Barwal, Aavasna Rupakheti.